\documentclass[12pt,a4paper]{article}

\usepackage{amsmath, amsthm, amssymb,accents,mathrsfs,mathtools}
\usepackage{enumerate}
\usepackage{pdflscape}
\usepackage{tikz}
\usetikzlibrary{arrows.meta, positioning}

\usepackage[pdftex,plainpages=false,hypertexnames=false,pdfpagelabels]{hyperref}
\usepackage{xcolor}
\definecolor{dark-red}{rgb}{0.7,0.25,0.25}
\definecolor{dark-blue}{rgb}{0.15,0.15,0.55}
\definecolor{medium-blue}{rgb}{0,0,.8}
\definecolor{DarkGreen}{RGB}{0,150,0}
\definecolor{rho}{named}{red}
\hypersetup{
   colorlinks, linkcolor={purple},
   citecolor={medium-blue}, urlcolor={medium-blue}
}

\usepackage{fullpage}
\theoremstyle{plain}
\newtheorem{thm}{Theorem}[section]
\newtheorem{thmalpha}{Theorem}

\newtheorem{coralpha}[thmalpha]{Corollary}
\newtheorem{cor}[thm]{Corollary}

\newtheorem{lem}[thm]{Lemma}
\newtheorem{prop}[thm]{Proposition}

\newtheorem*{thm*}{Theorem}
\theoremstyle{definition}
\newtheorem{defn}[thm]{Definition}
\newtheorem{ques}[thm]{Question}
\theoremstyle{remark}

\newtheorem{rem}[thm]{Remark}

\numberwithin{equation}{section}

\DeclareMathOperator{\End}{End}

\DeclareMathOperator{\Lie}{Lie}

\DeclareMathOperator{\Mob}{M\ddot{o}b}

\DeclareMathOperator{\supp}{supp}

\renewcommand{\Re}{\operatorname{Re}}
\renewcommand{\Im}{\operatorname{Im}}

\newcommand{\abs}[1]{\left| #1 \right|}
\newcommand{\ip}[1]{\langle #1 \rangle}

\newcommand{\norm}[1]{\left\| #1 \right\|}

\def\semicolon{;}
\def\applytolist#1{
    \expandafter\def\csname multi#1\endcsname##1{
        \def\multiack{##1}\ifx\multiack\semicolon
            \def\next{\relax}
        \else
            \csname #1\endcsname{##1}
            \def\next{\csname multi#1\endcsname}
        \fi
        \next}
    \csname multi#1\endcsname}

\def\calc#1{\expandafter\def\csname c#1\endcsname{{\mathcal #1}}}
\applytolist{calc}QWERTYUIOPLKJHGFDSAZXCVBNM;
\def\bbc#1{\expandafter\def\csname bb#1\endcsname{{\mathbb #1}}}
\applytolist{bbc}QWERTYUIOPLKJHGFDSAZXCVBNM;
\def\bfc#1{\expandafter\def\csname bf#1\endcsname{{\mathbf #1}}}
\applytolist{bfc}QWERTYUIOPLKJHGFDSAZXCVBNM;
\def\sfc#1{\expandafter\def\csname s#1\endcsname{{\sf #1}}}
\applytolist{sfc}QWERTYUIOPLKJHGFDSAZXCVBNM;

\newcommand{\noshow}[1]{}

\usepackage{wasysym}

\title{The Bisognano-Wichmann property for non-unitary Wightman conformal field theories}
\author{James E. Tener}
\date{}

\begin{document}

\maketitle

\begin{abstract}
The Bisognano-Wichmann and Haag duality properties for algebraic quantum field theories are often studied using the powerful tools of Tomita-Takesaki modular theory for nets of operator algebras.
In this article, we study analogous properties of nets of algebras generated by smeared Wightman fields, for potentially non-unitary theories.
In light of recent work \cite{CarpiRaymondTanimotoTener25} constructing Wightman field theories for (non-unitary) M\"obius vertex algebras, we obtain a broadly applicable non-unitary version of the Bisognano-Wichmann property.
In this setting we do not have access to the traditional tools of Hilbert space functional analysis, like functional calculus. Instead, results analogous to those of Tomita-Takesaki theory are derived `by hand' from the Wightman axioms.
As an application, we demonstrate Haag duality for nets of smeared Wightman fields.
\end{abstract}

{\hypersetup{linkcolor=black}
\setcounter{tocdepth}{2}
\tableofcontents
}

\section{Introduction}
The foundational articles \cite{BisognanoWichmann75,BisognanoWichmann76} exhibited what is now called the ``Bisognano-Wichmann property'' in axiomatic quantum field theory, which links 1) the analytic continuation of the symmetry group, 2) the PCT operator, and 3) the adjoint operation on smeared fields or local observables. 
Bisognano and Wichmann worked in the context of $4$-dimensional Wightman-type QFTs on Minkowski space, but analogous results were established for $2$-dimensional conformal field theories.
This was done under Wightman-type axioms in \cite{BuchholzSchulz-Mirbach90}, and using purely operator algebraic methods in \cite{BrunettiGuidoLongo93,GabbianiFrohlich93} under Haag-Kastler axioms.
In the operator algebraic setting, a beautiful picture emerges where physical symmetries realise the Tomita-Takesaki modular automorphisms and modular conjugation corresponding to the vacuum state.

In this article, we study the Bisognano-Wichmann property for 2d chiral conformal field theories (CFTs) which are potentially \emph{non-unitary}, which is to say theories without a compatible inner product.
We are motivated by the following question: can the operator algebraic methods of algebraic conformal field theory be applied to non-unitary conformal field theories?
While a positive answer would obviously be quite desirable due to the recent surge in interest in non-unitary CFTs, it is a daunting challenge to pursue this line of thought due to the fundamental reliance of operator algebraic methods on Hilbert space techniques.
We see a potential path forward as follows.

It has long been understood that the axioms of a vertex algebra (the most common way of axiomatizing a non-unitary conformal field theory) are closely tied to the Wightman axioms \cite[\S1.2]{Kac98}, and in recent joint work we showed that there is an equivalence of categories between (potentially non-unitary) M\"obius-covariant Wightman CFTs on the unit circle $S^1$ and (potentially non-unitary) M\"obius vertex algebras which are generated by a chosen set of quasiprimary vectors \cite{CarpiRaymondTanimotoTener25,RaymondTanimotoTener22}.
Such a Wightman theory consists of a family $\cF$ of operator-valued distributions $\varphi:C^\infty(S^1) \to \End(\cD)$ on a common, invariant domain $\cD$ cyclically generated from a vacuum vector $\Omega$, along with a compatible positive-energy representation $U$ of the M\"obius group $\Mob=\operatorname{PSU}(1,1)$ of holomorphic automorphisms of the unit disk.

Following \cite{StreaterWightman64}, we may then form for each interval\footnote{i.e. open connected nonempty non-dense subset} $I \subset S^1$ the polynomial algebra $\cP(I)$ generated by all smeared fields $\varphi(f)$, where $\varphi \in \cF$ and $f \in C^\infty(S^1)$ is supported in $I$.
These polynomial algebras formally satisfy many of the same axioms as a Haag-Kastler net of operator algebras, but they are actually built from a vertex algebra.
Our philosophy therefore is to try to obtain analogous results for the polynomial nets $\cP$ as one does with Haag-Kastler nets, and in doing so apply operator algebraic methods in the non-unitary setting.
An early success in this regard was the proof of the Reeh-Schlieder property that $\Omega$ is cyclic and separating for the algebras $\cP(I)$ for non-unitary Wightman theories \cite[App. A]{CarpiRaymondTanimotoTener25}.

In practice, transferring operator algebraic methods to the non-unitary setting can be quite difficult as these methods often rely crucially on features of operators on Hilbert spaces, e.g. by applying functional calculus to produce operators with certain properties.
However, in this article we report success in this endeavor with regard to the Bisognano-Wichmann property and Haag duality.
We are guided by results proven in the setting of unitary CFTs, but often a novel approach is required due to the lack of general functional analytic theory.
For example, one cannot invoke the general Tomita-Takesaki theory, but instead the key results can be established `by hand' through careful study of the CFT axioms; the same is true with regard to applications of functional calculus.

We now summarize some of the main results.
Fix a M\"obius-covariant Wightman CFT $(\cF,\cD,U,\Omega)$ on $S^1$ as described above.
Let 
\[
v_t(z) = \frac{\cosh(t/2)z - \sinh(t/2)}{-\sinh(t/2)z + \cosh(t/2)}
\]
be the one-parameter subgroup of $\Mob$ leaving the upper- and lower-semicircles $I_{\pm}$ invariant (which uniquely determines $v_t$ up to rescaling the parameter $t$), and let $V_t = U(v_t) \in \End(\cD)$.
The analytic continuations $\widetilde V_{\pm i \pi}$ have domains $D(\widetilde V_{\pm i\pi}) \subset \cD$ given by vectors $\Phi \in \cD$ for which there exists a $\Psi \in \cD$ such that the following holds: Let $\bbS_{\pm \pi i} \subset \bbC$ be the closed strip bounded by the lines $\bbR$ and $\bbR \pm \pi i$. For every continuous (with respect to the canonical $\cF$-strong topology on $\cD$) linear functional $\lambda:\cD \to \bbC$, there exists a function $G_\lambda:\bbS_{\pm \pi i} \to \bbC$, continuous on the closed strip and holomorphic on the interior, such that 
\begin{equation}\label{eqn: intro strip function}
G_\lambda(t) = \lambda(V_t\Phi), \qquad G_\lambda(t \pm i \pi) = \lambda(V_t \Psi)
\end{equation}
for all $t \in \bbR$.
Given such a $\Phi \in D(\widetilde V_{\pm i \pi})$ with its companion $\Psi \in \cD$, we set $\widetilde V_{\pm i \pi} \Phi = \Psi$.
Note that in contrast to the unitary setting, the operators $\widetilde V_{\pm i \pi}$ are not provided by functional calculus, and any properties of the operators must be established directly from the definition.

Traditionally, the Bisognano-Wichmann property has been formulated in a unitary setting, identifying the analytic continuation of the group $V_t$ with the PCT operator and the adjoint operation on smeared fields.
On the other hand, for a non-unitary CFT with no extra structure, one may not have a PCT operator or adjoint operation.
Even in this general setting, however, we are still able to compute the analytic continuation of the group $V_t$ (Theorem~\ref{thm: nonunitary Vipi} in the main body of the article):
\begin{thmalpha}\label{thm: intro nonunitary Vipi}
    Let $(\cF,\cD,U,\Omega)$ be a potentially non-unitary Wightman CFT on $S^1$. Then $\cP(I_+)\Omega \subset D(\widetilde V_{i\pi})$. If $\varphi_1, \ldots, \varphi_k \in \cF$ with conformal dimension $d_j$ and $f_j \in C^\infty(S^1)$ with $\operatorname{supp} f_j \subset I_+$, then 
    \[
    \widetilde V_{i\pi} \varphi_1(f_1) \cdots \varphi_k(f_k)\Omega = (-1)^{\sum d_j} \varphi_k(f_k \circ z^{-1}) \cdots \varphi_1(f_1 \circ z^{-1})\Omega.
    \]
    The same holds with $I_+$ replaced by $I_-$ and $\widetilde V_{i\pi}$ replaced by $\widetilde V_{-i\pi}$.
\end{thmalpha}
It is not too difficult to establish the result of Theorem~\ref{thm: intro nonunitary Vipi} when the functions $f_j$ lie in a dense subspace, but in the absence of the usual Hilbert space tools,   extending this result to all of $\cP(I_\pm)\Omega$ requires a careful study of the analytic functions \eqref{eqn: intro strip function}.
We do not assume anything like finite-dimensionality of the conformal weight spaces, and so Theorem~\ref{thm: intro nonunitary Vipi} applies to a variety of interesting log CFTs, such as the $\beta\gamma$-ghost system with $c=2$ (typically studied as a vertex algebra, but realizable as a Wightman CFT by \cite{CarpiRaymondTanimotoTener25}).

In order to say more, we will now assume that our Wightman CFT possesses an invariant nondegenerate Hermitian form $\ip{ \, \cdot \, , \, \cdot \, }$ on $\cD$, along with its antiunitary PCT operator $\theta$.
Such a theory would be unitary if we assumed the form to be positive definite, but there are important non-unitary examples motivating our study, such as the non-unitary minimal models (e.g. the Yang-Lee model).
In this case we have a notion of adjoint of an operator on $\cD$ (although the adjoint is not guaranteed to exist), and the polynomial algebras become $*$-algebras.
In this setting, we consider a slightly modified version of the operators $\widetilde V_{\pm i\pi}$, which we denote $V_{\pm i\pi}$, for which we only require analytic continuations \eqref{eqn: intro strip function} to exist for linear functionals $\lambda$ of the form $\ip{ \, \cdot \, , \Phi}$.
These modified operators are self-adjoint:
\[
V_{\pm i \pi}^* = V_{\pm i\pi}
\]
in the strong sense of partially defined operators on $\cD$ (Theorem~\ref{thm: Vipi selfadjoint}).
For these operators we obtain the familiar Bisognano-Wichmann property linking the analytic continuation of $V_t$, the PCT operator, and the adjoint operation; it appears as  Theorem~\ref{thm: involutive BW} in the body of the text.

\begin{thmalpha}\label{thm: into involutive BW}
    Let $(\cF,\cD,U,\Omega)$ be a potentially non-unitary Wightman CFT on $S^1$ with invariant nondegenerate Hermitian form and antiunitary PCT operator $\theta$.
    Then $\cP(I_\pm)\Omega \subset D(V_{\pm i\pi})$ and for $x \in \cP(I_\pm)$ we have
    \[
    V_{\pm i\pi} x\Omega = \theta x^*\Omega.
    \]
\end{thmalpha}

\noindent  We also obtain the KMS property as Corollary~\ref{cor: KMS}:
\begin{coralpha}
    Let $\phi:\cP(I_+) \to \bbC$ be the linear functional
    \[
    \phi(x) = \ip{x\Omega,\Omega}
    \]
    and let $\alpha_t = \operatorname{Ad} V_t \in \operatorname{Aut}(\cP(I_+))$.
    Let $\bbS_{2\pi i}$ be the closed strip bounded by $\bbR$ and $\bbR + 2\pi i$.
    Then for every $x,y \in \cP(I_+)$, there exists a function $F:\bbS_{2\pi i} \to \bbC$ that is continuous on the closed strip, holomorphic in the interior, and satisfies for all $t \in \bbR$
    \[
    F(t) = \phi(x\alpha_t(y)), \qquad F(t+2\pi i) = \phi(\alpha_t(y)x).
    \]
\end{coralpha}

As in the motivating articles \cite{BisognanoWichmann75,BuchholzSchulz-Mirbach90,BrunettiGuidoLongo93,GabbianiFrohlich93}, it is natural to next turn to questions of Haag duality.
Let $\End_*(\cD)$ be the $*$-algebra of continuous endomorphisms of $\cD$ with a continuous adjoint.
The algebras $\cP(I)$ are a non-unitary Haag-Kastler nets of subalgebras of $\End_*(\cD)$, as defined in Section~\ref{sec: nonunitary duality}.
This means that they satisfy: (i) isotony - $\cP(I) \subset \cP(J)$ when $I \subset J$, (ii) M\"obius covariance - they transform appropriately under a representation of $\Mob$, (iii) locality - $\cP(I)$ and $\cP(J)$ commute when $I$ and $J$ are disjoint, and (iv) vacuum - there is a chosen vector $\Omega$ which is cyclic for each algebra $\cP(I)$.
Let $\cQ(I)=\cP(I)''$, where the notion $S'$ refers to the commutant in $\End_*(\cD)$.
Then $\cQ(I)$ is again a Haag-Kastler net of subalgebras of $\End_*(\cD)$, and $\cP(I) \subseteq \cQ(I)$ for all intervals $I$.
The Haag duality property of $\cQ$ asserts that $\cQ(I') = \cQ(I)'$, where $I'$ denotes the (interior of the) complementary interval. 
We establish Haag duality  in Theorem~\ref{thm: algebra haag duality}.

\begin{thmalpha}\label{thm: intro duality}
Let $(\cF,\cD,U,\Omega)$ be a potentially non-unitary Wightman CFT on $S^1$ with invariant nondegenerate Hermitian form.
Then the algebras $\cQ(I)=\cP(I)''$ are a Haag dual M\"obius-covariant net of subalgebras of $\End_*(\cD)$.
That is, $\cQ(I') = \cQ(I)'$.
\end{thmalpha}

The importance of Theorem~\ref{thm: intro duality} lies in the ubiquitous role of Haag duality in the analysis of Haag-Kastler nets of operator algebras, particularly with regard to their representation theory.
In this context, the operation of `double commutant' is quite natural, as it provides the von Neumann algebra generated by a unital $*$-subalgebra of the bounded operators on a Hilbert space.
We hope that Theorem~\ref{thm: intro duality} opens up new avenues for applying Haag duality to the representation theory of non-unitary Wightman polynomial nets mimicking operator algebraic methods.

Let $\cP(I)_{sa}$ denote the real subspace of self-adjoint elements of $\cP(I)$.
The proof of Theorem~\ref{thm: intro duality} relies on studying the family of $\cF$-strongly closed real subspaces $\cK(I):=\overline{\cP(I)_{sa}\Omega} \subseteq \cD$.
These are a non-unitary analog of a net of standard subspaces in the sense of \cite{LongoLectureNotesI}.
Even in the unitary case, this result provides new avenues to apply operator algebraic methods to study unitary vertex operator algebras.

As a final application, we consider the question of ``AQFT-locality'' of unitary M\"obius vertex algebras.
Given such a vertex algebra, we can consider the Wightman CFT $(\cF,\cD,U,\Omega)$ corresponding to the generating set of all quasiprimary vectors.
The corresponding von Neumann algebras $\cA(I)$ are defined to be the von Neumann algebras generated by operators $\varphi(f)$\footnote{That is, the von Neumann algebra generated by the polar partial isometries and spectral projections of the closures of these operators.} with $\operatorname{supp} f \subset I$.
The unitary M\"obius vertex algebra is called AQFT-local if the algebras $\cA(I)$ and $\cA(J)$ commute when $I$ and $J$ are disjoint (this is closely related to the notion of ``strong locality'' of \cite{CKLW18}, but without the assumption of energy bounds).
In this case, the algebras $\cA(I)$ form a M\"obius covariant Haag-Kastler net by \cite[Prop. 4.8]{RaymondTanimotoTener22} (or \cite{CKLW18} in the case of energy bounded examples).
Using the net of standard subspaces associated with the Wightman CFT $\cF$, we obtain the following version of \cite[Thm. 3(f)]{BisognanoWichmann75} as Corollary~\ref{cor: separating implies local}:
\begin{thmalpha}
    Let $\cV$ be a unitary M\"obius vertex algebra with vacuum vector $\Omega$, and let $\cA(I)$ be the net of von Neumann algebras described above.
    Then $\cV$ is AQFT-local if and only if for some interval $I$ the vector $\Omega$ is a separating vector for $\cA(I)$.
\end{thmalpha}

This theorem shows that in order to verify AQFT-locality, it is enough to know that \emph{enough operators} commute with $\cA(I)$, not necessarily the operators in $\cA(J)$ (with $I$ and $J$ disjoint).

\subsection*{Dedication}

This article is dedicated to Huzihiro Araki, and was prepared for the special issue of \emph{Communications in Mathematical Physics} in his honor.
We briefly explain the connection between this article and Araki's work.
Araki's proof of Haag duality for a free scalar bosonic field \cite{Araki63,Araki64}, sometimes called Haag-Araki duality, laid the foundation for later study of duality using the methods of Bisognano-Wichmann;  this article is built on that foundation.
Araki also considered  duality for free fermionic fields \cite[Rem. 4.9]{Araki70}, and this instance of Haag-Araki duality plays a special role with regard to the philosophy of this article.
Indeed, we are motivated here by the comparison between nets $\cP(I)$ generated by (typically unbounded) smeared fields, and the corresponding nets of von Neumann algebras $\cA(I)$.
The free fermion model is essentially the only example of a 2d chiral CFT whose smeared fields are bounded operators, so that the distinction between $\cP(I)$ and $\cA(I)$ is not necessary.
Our Theorem~\ref{thm: intro duality} establishing Haag duality for nets of smeared fields can therefore be interpreted as a generalization of Haag-Araki duality to non-free models.

\subsection*{Acknowledgements}

The author was supported by ARC Discovery Project DP200100067.
The author would like to thank the anonymous referee for helpful suggestions.

\section{Background}\label{sec: background}

\subsection{Conformal nets and Tomita-Takesaki theory}\label{sec: background unitary}

Our first aim is to describe the relationship between Tomita-Takesaki modular theory and algebraic conformal field theories on Hilbert spaces which provides the motivation for this work.
We begin by summarizing the main results of Tomita-Takesaki theory, which was first introduced \cite{Tomita67,Takesaki70} in the context of von Neumann algebras (see also \cite{TakesakiTOAII}).
If $\cH$ is a Hilbert space, then a von Neumann algebra $M$ is a unital $*$-subalgebra $M \subset \cB(\cH)$ that is closed in the topology of pointwise convergence.
A vector $\Omega \in \cH$ is called cyclic for $M$ if $M\Omega$ is dense in $\cH$, and called separating if the only element $x \in M$ satisfying $x\Omega = 0$ is $x=0$.
A vector $\Omega$ is cyclic for $M$ if and only if it is separating for $M'$, and in light of the double commutant theorem $M=M''$, it follows that $\Omega$ is cyclic for $M'$ if and only if it is separating for $M$\footnote{Here the commutant $S'$ of a subset $S \subset \cB(\cH)$ consists of all operators $x \in \cB(\cH)$ commuting element-wise with $S$. This is a von Neumann algebra if $S$ is closed under taking adjoints}.

Let $S^0$ be the densely-defined conjugate linear operator on $\cH$ which maps $M\Omega$ to itself via the formula $S^0x\Omega = x^*\Omega$, and let $S$ be the closure of $S^0$ (i.e. the operator obtained by taking the closure of the graph of $S^0$ in $\cH \times \cH$).
We have a polar decomposition $S=J\Delta^{1/2}$, where $J$ is antiunitary and $\Delta^{1/2}$ is a densely defined positive operator.
The main result of Tomita-Takesaki theory  says that
\[
JMJ=M', \qquad \text{and} \qquad \Delta^{it}M\Delta^{-it} = M,
\]
where the unitary operator $\Delta^{it}$ is defined via spectral theory and functional calculus.
We thus have a canonical dynamics arising from the modular automorphism group $\alpha_t := \operatorname{Ad} \Delta^{it}$ (which is uniquely determined up to inner automorphisms by the work of Connes).

Tomita-Takesaki theory has proven to be an extremely powerful tool in the study of algebraic quantum field theories, which is to say quantum field theories in the sense of Haag-Kastler nets of operator algebras \cite{HaagKastler64}.
In the context of 2d chiral conformal field theory, we have the following axiomatization:
\begin{defn}\label{def: HK}
A \textbf{M\"obius-covariant (Haag-Kastler) net} on $S^1$ is a triple $(\cA, U, \Omega)$, where
$\cA$ associates to each open non-dense nonempty connected interval $I$ of $S^1$ a von Neumann algebra $\cA(I)$ on $\cH$,
$U$ is a SOT-continuous unitary representation of $\Mob$ on $\cH$, and $\Omega \in \cH$ such that the following hold:
\begin{enumerate}[{(HK}1{)}]
\item \textbf{Isotony}: If $I_1 \subset I_2$, then $\cA(I_1) \subset \cA(I_2)$.
\item \textbf{Locality}: If $I_1 \cap I_2 = \emptyset$, then $\cA(I_1)$ and $\cA(I_2)$ commute.
\item \textbf{M\"obius covariance}: For $\gamma \in \Mob$, $U(\gamma)\Omega = \Omega$ and
$U(\gamma)\cA(I) U(\gamma)^* = \cA(\gamma I)$ for each interval $I$.
\item \textbf{Spectrum condition}: The generator $L_0$ of rotations $U(r_\theta)=e^{i\theta L_0}$ is positive.
\item \textbf{Vacuum}: $\Omega$ is the unique (up to a scalar) vector in $\cH$ that is invariant under $U$,
and it is cyclic for $\bigvee_{I\Subset S^1}\cA(I)$.
\end{enumerate}
\end{defn}
The condition that $\Omega$ is the unique $\Mob$-invariant vector is sometimes omitted, in which case a net with this uniqueness property would be called an `irreducible' net.

As a consequence of the axioms, the Reeh-Schlieder property holds: the vector $\Omega$ is cyclic and separating for each local algebra $\cA(I)$.
One can then study the closed conjugate-linear involution $S_+$ associated with the algebra $\cA(I_+)$, where $I_+ = \{z \in S^1 \, : \, \Im(z) > 0\}$ is the upper semi-circle (and likewise with the lower semi-circle $I_-$).
The Bisognano-Wichmann property of the Haag-Kastler net identifies the modular automorphism group $\Delta_+^{it}$ of this algebra (with respect to the vacuum vector $\Omega$) with a certain subgroup of the M\"obius group, and gives a covariance relation for the modular conjugation $J_+$ of the same algebra.
Specifically consider the following one-parameter subgroup $v_t \in \Mob$ of transformations fixing the intervals $I_{\pm}$
\[
    v_t(z) = \frac{\cosh(t/2)z - \sinh(t/2)}{-\sinh(t/2)z + \cosh(t/2)},
\]
and let $V_t = U(v_t)$ be the corresponding unitary operator on $\cH$.
Then the Bisognano-Wichmann property says that the modular automorphism group and modular conjugation satisfy
\begin{equation}\label{eqn: operator algebra BW}
\Delta_+^{it} = V_{-2\pi t}, \qquad J_+\cA(I_+)J_+ = \cA(I_-).
\end{equation}
The original work of Bisognano-Wichmann \cite{BisognanoWichmann75,BisognanoWichmann76} featured a computation of the analytic continuation of $V_t$ in the context of Wightman quantum field theories on $4$-dimensional Minkowski space.
Analogous results in the Wightman setting were obtained for conformal field theories by Buchholz-Schulz-Mirbach \cite{BuchholzSchulz-Mirbach90}, and purely operator algebraic proofs of \eqref{eqn: operator algebra BW} for conformal field theories were given separately in \cite{BrunettiGuidoLongo93,GabbianiFrohlich93}.

In all of the above investigations of the Bisognano-Wichmann property, one of the key aims was to study Haag duality, which is the fact that in any Haag-Kastler net we have
\[
  \cA(I)' = \cA(I').  
\]
This fact was proven again separately in \cite{BrunettiGuidoLongo93,GabbianiFrohlich93}, in both cases invoking work of Borchers \cite{Borchers92}.
However, for our purposes it will be enlightening to briefly describe an alternate approach via standard `one-particle' subspaces, due to Longo \cite{LongoLectureNotesI,LongoLectureNotesII}.

The one-particle approach relies on the approach to Tomita-Takesaki developed in \cite{RieffelVanDaele77}, in which the authors discovered that the modular operators $J$ and $\Delta$ of a von Neumann algebra $M \subset \cB(\cH)$ with cyclic and separating vector $\Omega$ could be recovered entirely from the real subspace $\cK:=\overline{M_{sa}\Omega} \subset \cH$, where $M_{sa}$ is the real subspace of self-adjoint elements.

\begin{defn}
    A \textbf{standard subspace} of a complex Hilbert space $\cH$ is a real subspace $\cK$ such that $\cK + i \cK$ is dense in $\cH$ and $\cK \cap i \cK = \{0\}$.
\end{defn}

The subspace $\overline{M_{sa}\Omega}$ is an example of a standard subspace, but many results of Tomita-Takesaki theory can be generalized to the setting of arbitrary standard subspaces (a fact which is crucial for the approach taken in this article).
Associated to $\cK$ one has a closed conjugate-linear involution $S$ defined by $S(\xi + i\eta) = \xi - i\eta$ for $\xi,\eta \in \cK$, and again one defines an antiunitary $J$ and positive operator $\Delta$ via the polar decomposition $S=J\Delta^{1/2}$.
Similar to the von Neumann setting, we have
\[
J\cK = \cK', \qquad \text{and} \qquad \Delta^{it}\cK=\cK
\]
where the \textbf{complement} $\cK'$ is defined by
\[
\cK' = \{ \xi \in \cH \, : \, \ip{\xi,\eta} \in \bbR \, \text{ for all } \, \eta \in \cK\}.
\]
Alternatively, $\cK' = (i\cK)^{\perp}$ where the orthogonal complement is taken with respect to the real Hilbert space inner product $\Re \ip{\, \cdot \, , \, \cdot \,}$ (from which it follows that $\cK'' = \cK$).
Following \cite{LongoLectureNotesI}, one has the following notion of a net of standard subspaces:
\begin{defn}\label{defn: unitary net of subspaces}
    A \textbf{M\"obius-covariant net of subspaces} of a complex Hilbert space $\cH$ is a family of real subspaces $\cK(I) \subset \cH$ indexed by open non-dense nonempty intervals $I \subset S^1$, and a unitary representation $U$ of $\Mob$ on $\cH$, satisfying the following properties.
    \begin{enumerate}[{(NS}1{)}]
\item \textbf{Isotony}: If $I_1 \subset I_2$, then $\cK(I_1) \subset \cK(I_2)$.
\item \textbf{Locality}: If $I_1 \cap I_2 = \emptyset$, then $\cK(I_1) \subset \cK(I_2)'$.
\item \textbf{M\"obius covariance}: For $\gamma \in \Mob$, we have $U(\gamma)\cK(I) = \cK(\gamma(I))$.
\item \textbf{Spectrum condition}: The generator $L_0$ of rotations $U(r_\theta)=e^{i\theta L_0}$ is positive.
\item \textbf{Cyclicity}: The complex linear span of the $\cK(I)$ is dense in $\cH$.
\end{enumerate}
\end{defn}
Given a Haag-Kastler net of algebras $\cA(I)$, there is a corresponding `one-particle' net of subspaces $\cK(I)=\overline{\cA(I)_{sa}\Omega}$.
For any net of subspaces, we again have a Reeh-Schlieder property, that each $\cK(I)$ is a cyclic subspace, and indeed each $\cK(I)$ is a standard subspace of $\cH$.
The Bisognano-Wichmann property for the net of subspaces says that we again have
\[
\Delta_+^{it} = V_{-2\pi t}, \qquad J_+\cK(I_+)=\cK(I_-),
\]
from which we again have Haag duality:
\[
\cK(I)' = \cK(I').
\]
If the net of subspaces $\cK(I)$ arises from a Haag-Kastler net of algebras of $\cA(I)$, then the Bisognano-Wichmann and Haag duality properties of $\cA$ can be derived from those of $\cK$ \cite{LongoLectureNotesII}.
In this article, we will expand on this connection in the case of non-unitary conformal field theories, which is to say conformal field theories whose state space does not possess an invariant positive-definite form. We will also construct nets of standard subspaces from Wightman theories (or, equivalently, from M\"obius vertex algebras), and use this to study the relationship between Wightman theories and Haag-Kastler nets.

\subsection{Wightman CFTs and vertex algebras}\label{sec: background wightman and va}

Our main object of study is potentially non-unitary M\"obius-covariant Wightman conformal field theories on the unit circle $S^1 \subset \bbC$.
In this context, the M\"obius group $\Mob$ refers to $\operatorname{PSU}(1,1)$, the group of holomorphic automorphisms of the unit disk.
Such theories were systematically studied in \cite{CarpiRaymondTanimotoTener25} (see also \cite{RaymondTanimotoTener22}).
In contrast to Section~\ref{sec: background unitary}, we do not limit ourselves to unitary theories (which possess a compatible inner product).
We will summarize here the basic notions, and the reader may consult the reference for further details.

A M\"obius covariant Wightman CFT on $S^1$ is given by a common domain $\cD$, which begins simply as a (generally infinite-dimensional) complex vector space, along with a family $\cF$ of operator-valued distributions $\varphi:C^\infty(S^1) \to \cL(\cD)$, where $\cL(\cD)$ denotes linear operators on $\cD$.
We will need to discuss one minor technical topological consideration for such a family $\cF$.

\begin{defn}
A linear functional $\lambda:\cD \to \bbC$ is called \textbf{compatible with $\cF$} if the multilinear maps $C^\infty(S^1)^k \to \bbC$ given by
\begin{equation*}\label{eqn: preliminary regular action}
(f_1, \ldots, f_k) \mapsto \lambda\big(\varphi_1(f_1) \cdots \varphi_k(f_k)\Phi\big)
\end{equation*}
are (jointly) continuous in the $f_j$ for all $\varphi_1, \ldots, \varphi_k \in \cF$ and $\Phi \in \cD$.
We write $\cD_\cF^*$ for the space of all linear functionals compatible with $\cF$.
We say that $\cF$ acts \textbf{regularly} on $\cD$ if $\cD_\cF^*$ separates points, which is to say that for every non-zero $\Phi \in \cD$ there exists $\lambda \in \cD_\cF^*$ such that $\lambda(\Phi)\ne0$.
The \textbf{$\cF$-strong topology} on $\cD$ is the finest locally convex topology such that expressions
$
\varphi_1(f_1) \cdots \varphi_k(f_k)\Phi
$
are jointly continuous in the functions $f_j \in C^\infty(S^1)$.
\end{defn}
We use the Fr\'echet topology induced on $C^\infty(S^1)$ by the $C^N$ norms\footnote{In \cite{RaymondTanimotoTener22,CarpiRaymondTanimotoTener25} we used the Sobolev norms $\norm{f}_{H^N}$, but the resulting topologies are equivalent due to embeddings $H^N(S^1) \subset C^{N'}(S^1)$ and $C^{M}(S^1) \subset H^{M'}(S^1)$.}
\[
\norm{f}_{C^N} = \sum_{j=0}^N \big\|f^{(j)}\big\|_\infty = \sum_{j=0}^N \big\|\tfrac{d^j}{d\theta^j} f\big\|_\infty.
\]

The continuous dual of $\cD$ (with respect to the $\cF$-strong topology) is exactly $\cD_\cF^*$, and so by the Hahn-Banach theorem the $\cF$-strong topology is Hausdorff precisely when the action of $\cF$ is regular.
For the remainder of the article, all locally convex vector topologies will be assumed Hausdorff.
Given such a family $\cF$, we will equip $\cD$ with the $\cF$-strong topology unless stated otherwise.

We now assume that in addition to the operator-valued distributions $\cF$, we have a representation $U:\Mob \to \End(\cD)$ of the M\"obius group, where $\End(\cD)$ denotes ($\cF$-strong) continuous linear operators.
For $\gamma \in \Mob$ we denote by $X_\gamma \in C^\infty(S^1)$ the function
\begin{equation}\label{eqn: X gamma}
X_\gamma(e^{i \theta}) = -i \frac{d}{d \theta} \log(\gamma(e^{i \theta})).
\end{equation}
For $f \in C^\infty(S^1)$ and $d \in \bbZ_{\ge 0}$ we denote by $\beta_d(\gamma)f \in C^\infty(S^1)$ the function
\begin{equation*}\label{eqn: beta}
(\beta_d(\gamma)f)(z) = (X_\gamma(\gamma^{-1}(z)))^{d-1} f(\gamma^{-1}(z)).
\end{equation*}
An operator-valued distribution with domain $\cD$ is called \textbf{M\"obius-covariant with conformal dimension $d$} under the representation $U$
if for every $\gamma \in \Mob$ and every $f \in C^\infty(S^1)$ we have 
\[
U(\gamma)\varphi(f)U(\gamma)^{-1} = \varphi(\beta_d(\gamma)f)
\]
as endomorphisms of $\cD$.
We say that a vector $\Phi \in \cD$ has \textbf{conformal dimension $d \in \bbZ$} if $U(r_\theta)\Phi = e^{i d\theta}\Phi$ for all rotations $r_\theta \in \Mob$.

\begin{defn}\label{def: Wightman}
Let $\cD$ be a vector space equipped with a representation $U$ of $\Mob$ and a choice of non-zero vector $\Omega \in \cD$.
Let $\mathcal{F}$ be a set of operator-valued distributions on $S^1$ acting regularly on their common domain $\mathcal{D}$.
This data forms a \textbf{M\"obius-covariant Wightman CFT} on $S^1$ if they satisfy the following axioms:
\begin{enumerate}[{(W}1{)}]
\item \textbf{M\"obius covariance}:  For each $\varphi \in \mathcal{F}$ there is $d \in \mathbb{Z}_{\ge 0}$ such that $\varphi$ is M\"obius-covariant with conformal dimension $d$ under the representation $U$.\label{itm: W Mob covariance}
\item \textbf{Locality}: If $f$ and $g$ have disjoint supports, then $\varphi_1(f)$ and $\varphi_2(g)$ commute for any pair $\varphi_1, \varphi_2 \in \mathcal{F}$.
\label{itm: W locality}
\item \textbf{Spectrum condition}: If $\Phi \in \cD$ has conformal dimension $d < 0$ then $\Phi = 0$.
\label{itm: W spec}
\item \textbf{Vacuum}: The vector $\Omega$ is invariant under $U$, and $\cD$ is algebraically spanned by vectors of the form $\varphi_1(f_1) \cdots \varphi_k(f_k)\Omega$.
\label{itm: W Vacuum}
\end{enumerate}
\end{defn}

The Wightman axioms have not been extensively studied in their own right within the context of axiomatic 2d chiral conformal field theory, but it turns out that they are exactly equivalent to M\"obius vertex algebras, along with a chosen set of quasiprimary vectors which generate the vertex algebra.
It has been long-known that this correspondence is ``morally valid'' (see e.g. \cite[\S1.2]{Kac98}), certainly in the case of unitary theories, but it was made precise in \cite{CarpiRaymondTanimotoTener25}.
In the case of vertex algebras, the M\"obius covariance condition is with respect to the complexified Lie algebra $\Lie(\Mob)_\bbC$, which is spanned by elements $\{L_{-1},L_0,L_1\}$ satisfying the familiar relation $[L_n,L_m] = (n-m)L_{n+m}$.

\begin{defn}
An ($\bbN$-graded) \textbf{M\"obius vertex algebra} consists of a vector space $\cV$ equipped with a representation $\{L_{-1},L_0,L_1\}$ of $\Lie(\Mob)_\bbC$, a state-field correspondence $Y:\cV \to \End(\cV)[[z^{\pm 1}]]$,
and a choice of non-zero vector $\Omega \in \cV$ such that the following hold:
\begin{enumerate}[{(VA}1{)}]
\item $\cV = \bigoplus_{n=0}^\infty \cV(n)$, where $\cV(n) =  \ker(L_0 - n)$.
\item $Y(\Omega,z) = \operatorname{Id}_\cV$ and $\left.Y(v,z)\Omega\right|_{z=0} = v$, i.e. $Y(v,z)\Omega$ has only non-negative powers of $z$ for all $v \in \cV$.
\item $\Omega$ is $\Lie(\Mob)$-invariant, i.e.\! $L_m \Omega = 0$ for $m=-1,0,1$.
\item $[L_m, Y(v,z)] = \sum_{j=0}^{m+1} \binom{m+1}{j}z^{m+1-j} Y(L_{j-1}v,z)$ and $Y(L_{-1}v,z) = \frac{d}{dz} Y(v,z)$ for all $v \in \cV$ and $m = -1,0,1$.
\item For each $u,v \in \cV$, there exists $N$ sufficiently large such that \newline $(z-w)^N[Y(v,z), Y(u,w)] = 0$.
\end{enumerate}
\end{defn}

For every Wightman theory $\cF$ with domain $\cD$, there exists a canonical M\"obius vertex algebra structure on the finite-energy vectors $\cV$, i.e. the subspace spanned by vectors of the form $\varphi_1(e_{j_1}) \cdots \varphi_k(e_{j_k})\Omega$, where $e_j(z) = z^j$.
The vertex algebra structure is characterized by the fact that for every $\varphi \in \cF$ with conformal dimension $d$ there is a vector $v_\varphi \in \cV(d)$ such that $\varphi(e_j) = v_{(d-j-1)}$.
The vectors $\{v_{\varphi} \, : \, \varphi \in \cF\}$ generate $\cV$ as a vertex algebra, and they are quasiprimary (i.e. $L_1 v_\varphi = 0$ and $v_\varphi$ is homogeneous for $L_0$).

In the other direction, given a M\"obius vertex algebra $\cV$ with a generating set of quasiprimary vectors, there is a canonical extension $\cD \supset \cV$ (which can be realized as a subspace of the algebraic completion of $\cV$) and M\"obius-covariant Wightman CFT on $\cD$ corresponding to $\cV$ in the sense described above.
This correspondence can be made into an equivalence of categories \cite[\S 3.3]{CarpiRaymondTanimotoTener25}.

Note that we do not require either of the conditions $\dim \cV(n) < \infty$ or $\dim \cV(0) = 1$.
The requirement that the vertex algebra $\cV$ be generated by quasiprimary vectors is not very restrictive, and the class of theories considered includes log-CFTs like the $\beta \gamma$-ghost system with $c=2$.

In both settings, we have a notion of an invariant Hermitian form, which we call an \textbf{involutive} theory in general, and a \textbf{unitary} theory if the Hermitian form is an inner product.
Some key examples of involutive but non-unitary theories include the non-unitary minimal models, such as the Yang-Lee model.
These notions will be outlined in detail in Sections~\ref{sec: sesquilinear} and \ref{sec: unitary}, respectively.
Invariant forms for a Wightman theory exactly correspond to invariant forms for the corresponding vertex algebra \cite[\S4]{CarpiRaymondTanimotoTener25}.

We close with a brief preview of the themes of the article.
Following \cite{StreaterWightman64}, given a Wightman CFT $\cF$ acting on its domain $\cD$, for every interval\footnote{The term `interval' in $S^1$ will be used exclusively for open, nonempty, proper, non-dense connected subsets of $S^1$.} $I \subset S^1$ we form the polynomial algebra $\cP(I) \subset \End(\cD)$, which is to say the unital algebra generated by operators $\varphi(f)$ with $\supp(f) \subset I$.
We obtain a net $\cP$  of algebras which formally satisfies many of the properties of a Haag-Kastler net, but is in fact encoding the data of a M\"obius vertex algebra.
From \cite[App. A]{CarpiRaymondTanimotoTener25}, this net has the Reeh-Schlieder property, which is to say that $\cP(I)\Omega$ is ($\cF$-strong) dense in $\cD$ for every interval $I$, and if $x \in \cP(I)$ and $x\Omega = 0$, then $x = 0$.
Continuing along the same line, we have the following motivating questions:

\begin{enumerate}
    \item Given that the tools of operator algebras have been so successful in studying Haag-Kastler nets, can any of the same tools be applied to the polynomial net $\cP$ of a (potentially non-unitary) Wightman theory?
    \item Given the extensive work that has been done to compare unitary vertex operator algebras and Haag-Kastler nets, can structural properties of polynomial nets be used as a bridge between these notions?
\end{enumerate}

In the remaining sections we develop some promising results in the direction of these two questions.

\section{Analytic continuation for non-unitary Wightman CFTs}\label{sec: general nonunitary}

Fix a non-unitary M\"obius-covariant Wightman conformal field theory $(\cF,\cD,U,\Omega)$ on the unit circle $S^1$ (see Definition~\ref{def: Wightman}).
This is precisely equivalent to the data of a M\"obius vertex algebra and a choice of generating family of quasiprimary states (see Section~\ref{sec: background wightman and va}).

Let $v_t \in \Mob$ be the one-parameter group
\[
    v_t(z) = \frac{\cosh(t/2)z - \sinh(t/2)}{-\sinh(t/2)z + \cosh(t/2)}.
\]
Let $I_\pm \subset S^1$ be the upper- and lower-semicircles, namely
\[
I_{\pm} = \{ z \in S^1 \, : \, \pm \Im(z) > 0\}.
\]
Then $v_t$ leaves each interval $I_\pm$ invariant, and is the unique one-parameter subgroup of $\Mob = \operatorname{PSU}(1,1)$ with this property, up to rescaling $t$.
Let $V_t = U(v_t) \in \End(\cD)$.
The M\"obius covariance axiom says that for any $\varphi \in \cF$ with conformal dimension $d$ and $f \in C^\infty(S^1)$ we have
\[
V_t \varphi(f) V_t^{-1} = \varphi(\beta_d(v_t)f)
\]
where
\[
\beta_d(\gamma)f = (X_\gamma^{d-1} \cdot f) \circ \gamma^{-1}
\]
and $X_{\gamma}$ is defined in \eqref{eqn: X gamma}.
Explicitly,
\[
X_{v_t}(v_{-t}(z)) = \cosh(t) + \Re(z)\sinh(t),
\]
which extends to an entire function of $t$.

For an interval $I \subset S^1$ we let $\cP(I) \subset \End(\cD)$ be the unital algebra generated by operators $\varphi(f)$ with $\varphi \in \cF$ and $\supp f \subset I$.
For $\tau \in \bbC$, let $\bbS_\tau$ be the closed strip bounded by the lines $\bbR$ and $\bbR + \tau$.
Our first goal is to investigate a certain weak form of analytic continuation of $V_t$ to the closed strip
\[
\bbS := \bbS_{i\pi} = \{\tau \in \bbC \, : \, 0 \le \Im(\tau) \le \pi\}.
\]
In this section, we consider the analytic continuation of $V_t$ when the CFT is not assumed to have any additional structure; in particular, for the moment we do not assume the existence of an invariant nondegenerate Hermitian form.
A key example of this type is the $\beta \gamma$-ghost system with central charge $c=2$ (along with other ``$A$-graded'' log-CFTs).
Recall (Section~\ref{sec: background wightman and va}) that we equip $\cD$ with its $\cF$-strong topology, and write $\cD_\cF^*$ for the continuous dual, consisting of linear functionals $\lambda:\cD \to \bbC$ such that the maps
\[
(f_1, \ldots, f_k) \mapsto \lambda\big(\varphi_1(f_1) \cdots \varphi_k(f_k)\Phi\big)
\]
are (jointly) continuous in the $f_j \in C^\infty(S^1)$, for every $\Phi \in \cD$.

Our goal is to compute the analytic continuation of $V_t$ on the domains $\cP(I_\pm)\Omega$.
We will generally focus on the domain $\cP(I_+)\Omega$, with the case of $\cP(I_-)\Omega$ being similar.
There are numerous references carrying out analogous computations in different contexts (e.g. \cite{BisognanoWichmann75} or \cite{BuchholzSchulz-Mirbach90}), but these generally employ methods which are not suitable for our context, in particular due to use of functional calculus and spectral theory for unbounded operators.
The argument presented here is adapted from the one presented in \cite{Henriques2014CFTNotes} for WZW models, although significant modifications are still necessary to account for the lack of unitarity.

We  now define the key object of study in Section~\ref{sec: general nonunitary}, the analytically continued operators $\widetilde V_\tau$ for $\tau \in \bbC$.
These will be partially defined operators\footnote{We say `partially defined operator' instead of 'densely defined operator' at this point because we have not yet demonstrated that the operators are densely defined.} on $\cD$, meaning they consist of a pair: a domain $D(\widetilde V_\tau)$ and an operator $\widetilde V_\tau:D(\widetilde V_\tau) \to \cD$.
\begin{defn}\label{def: Vs}
For $\tau \in \bbC$, the operator $\widetilde V_\tau$ has domain $D(\widetilde V_\tau) \subset \cD$ consisting of vectors $\Phi \in \cD$ such that for every $\lambda \in \cD_\cF^*$ the function $\bbR \to \bbC$ given by
\[
t \mapsto \lambda\big( V_t \Phi\big)
\]
extends to a (necessarily unique) function $G_\lambda$ on the closed strip $\bbS_\tau$ bounded by $\bbR$ and $\bbR + \tau$ such that $G_{\lambda}$ is continuous on $\bbS_\tau$ and holomorphic in the interior $\mathring{\bbS_\tau}$. Moreover, we require that there exist $\Psi \in \cD$ such that $G_{\lambda}(\tau+t) = \lambda\big(V_t\Psi\big)$ for all $t \in \bbR$ and $\lambda \in \cD_\cF^*$.
The operator $\widetilde V_\tau$ is defined on $\Phi \in D(\widetilde V_\tau)$ by $\widetilde V_\tau\Phi = \Psi$.
\end{defn}
Note that such a vector $\Psi$ is uniquely determined by $\Phi$ (i.e. $\widetilde V_\tau$ is well-defined) by the regularity property of the Wightman CFT, which asserts that $\cD_\cF^*$ separates points. 

If $\tau=t \in \bbR$ then $D(\widetilde V_t) = \cD$ and $\widetilde V_t = V_t$.
The tilde in the notation $\widetilde V_\tau$ and $D(\widetilde V_\tau)$ is present to distinguish from the domain $D(V_\tau)$ considered in Section~\ref{sec: nonunitary duality} in the context of CFTs with an invariant Hermitian form.
If $\tau \in \bbC$ and $t \in \bbR$ we have $\widetilde V_{\tau+t} = \widetilde V_\tau V_t = V_t \widetilde V_\tau$ as partially defined operators\footnote{This means that the operators in question have the same domain, where the domain of the composite of partially defined operators $X_1X_2$ consists of all vectors $\Phi$ such that $\Phi \in D(X_2)$ and $X_2\Phi \in D(X_1)$.}.
More generally for $\tau_1,\tau_2 \in \bbC$ the operator $\widetilde V_{\tau_1+\tau_2}$ is an extension of $\widetilde V_{\tau_1}\widetilde V_{\tau_2}$.
It is clear that $\Omega \in D(\widetilde V_\tau)$, but it is not immediately obvious that $D(\widetilde V_\tau)$ is any larger.
Thus our first step is to identify a `large' subspace of $D(\widetilde V_\tau)$.

Let $C_0^\infty(I_+)$ and $C_0^\infty(I_-)$ be smooth functions on $I_+$ (resp. $I_-$) that vanish to all orders at the endpoints.
We will sometimes embed these spaces into $C^\infty(S^1)$ by extending with the value $0$, so that $C^\infty_0(I_\pm)$ may be identified with the kernel of the restriction map $C_0^\infty(S^1) \to C_0^\infty(I_\mp)$, where $C_0^\infty(S^1)$ is the space of smooth functions on $S^1$ that vanish to all orders at $\pm 1$.
Let 
\[
\bbO = \{z \in \bbC \, : \, \abs{z} \ge 1\} \cup \{\infty\}.
\]
The following class of functions will feature in our calculation below.
\begin{defn}\label{def: cX}
    We denote by $\cX$ the space of functions $F:\bbO \to \bbC$ which are smooth on $\bbO$, holomorphic in the interior $\mathring{\bbO}$, vanish at infinity, and vanish to all orders at $\pm 1$.
\end{defn}
We give $\cX$ the topology inherited from its embedding into $C^\infty(S^1)$ via restriction, namely the topology of uniform convergence of all derivatives on $S^1$, which we call the $C^\infty$ topology.
If $F \in \cX$ we will frequently consider $F|_{I_{+}} \in C_0^\infty(I_{+})$ for the restriction to the upper semicircle, and similarly for $I_-$.
We will typically extend $F|_{I_+}$ to an element of $C^\infty(S^1)$ as described above, by assigning the value $0$ to the interval $I_-$.
Thus we have
\[
F|_{S^1} = F|_{I_+} + F|_{I_-}.
\]

\begin{defn}\label{def: cPcX}
    Let $\cP_\cX(I_+) \subset \cP(I_+)$ be the unital subalgebra generated by operators $\varphi(F|_{I_+})$ with $\varphi \in \cF$ and $F \in \cX$.
    The subalgebra $\cP_\cX(I_-) \subset \cP(I_-)$ is defined similarly.
\end{defn}

We will show that $\cP_\cX(I_+)\Omega$ is a dense subspace of $\cD$ and is contained in $D(\widetilde V_{i\pi})$.
A key technical observation is the following, whose proof is adapted from \cite[p.41]{Henriques2014CFTNotes}:
\begin{lem}\label{lem: X dense}
The space $\{F|_{I_+} \, : \, F \in \cX\}$ is dense in $C^\infty_0(I_+)$, and similarly for $I_-$.
\end{lem}
\begin{proof}
By precomposing with the map $\tfrac{1}{z}$, we can consider functions holomorphic on the closed unit disk $\bbD=\{\abs{z} \le 1\}$ rather than $\bbO$, and without loss of generality we consider only the case of $I_+$.
Our problem is then: fix $f \in C_0^\infty(I_+)$, and find a sequence of functions $F_N$ on the closed unit disk that satisfy: (i) $F_N$ is smooth on $\bbD$ and $F_N(0)=0$, (ii) $F_N$ is holomorphic on the open unit disk $\mathring{\bbD}$, (iii) $F_N|_{I_+} \in C^\infty_0(I_+)$, and (iv) $F_N|_{I_+} \to f$ in the $C^\infty$ topology.
Without loss of generality we may assume that $f$ is identically zero in a neighborhood of $\{\pm 1\}$, as such functions are dense in $C^\infty_0(I_+)$.
We may also ignore the requirement $F_N(0)=0$, as if we approximate $z^{-1}f$ by $F_N$, then $f$ is approximated by $z F_N$.

Following \cite{Henriques2014CFTNotes} we first define $F_N \in C^\infty_0(I_+)$ as a convolution 
\begin{equation}\label{eqn: FN interval}
F_N|_{I_+} := \frac{N}{\sqrt{2\pi}} \int_{-\infty}^\infty e^{-\tfrac{(Ns)^2}{2}} \cdot f \circ v_{-s} \, ds.
\end{equation}
We have $F_N|_{I_+} \in C^\infty(I_+)$ and $F_N|_{I_+} \to f$ in the $C^\infty$ topology, which demonstrates points (iii) and (iv) provided we can show that $F_N$ extends to a (necessarily unique) function on the closed unit disk satisfying (i) and (ii).
We will do this in two steps.
In Step 1, we will show that $F_N$ extends from its original domain $I_+$ to a function holomorphic in a neighborhood of $\bbD \setminus \{\pm 1\}$.
In Step 2, we will show that $F_N$ extends further to take the value zero at $z=\pm 1$, and that the resulting function is smooth on the closed disk $\bbD$.
These two steps demonstrate the existence of the required sequence $F_N$, at which point the proof will be complete.

We begin with Step 1, and construct an analytic extension of $F_N$ to a neighborhood of $\bbD \setminus \{\pm 1\}$ as follows.
As before, let $\bbS=\bbS_{i\pi}$ be the strip 
\[
\bbS = \{ z \in \bbC \, : \, 0 \le \Im(z) \le \pi\}
\]
and let $\alpha:\bbD \setminus \{\pm 1\} \to \bbS$ be the holomorphic map
\[
\alpha(z) = \log\left(-i \frac{z-1}{z+1}\right)
\]
with the branch of $\log$ on the upper half-plane chosen so that the image of $\alpha$ is $\bbS$.
This map satisfies $\alpha(\mathring{I_+}) = \bbR$ and $\alpha(\mathring{I_-}) = \bbR + i \pi$.
In fact, for $z \in \mathring{I_+}$ we have $v_{\alpha(z)}(i) = z$, so that 
\begin{align*}
 F_N(z) &= \frac{N}{\sqrt{2\pi}} \int_{-\infty}^\infty e^{-\tfrac{(Ns)^2}{2}} \, f(v_{-s}(z))\,ds\\
  &= \frac{N}{\sqrt{2\pi}} \int_{-\infty}^\infty e^{-\tfrac{(Ns)^2}{2}} \, f(v_{-s+\alpha(z)}(i))\,ds\\
 &= \frac{N}{\sqrt{2\pi}} \int_{-\infty}^\infty e^{-\tfrac{(N(s+\alpha(z)))^2}{2}} \, f(v_{-s}(i))\,ds.
\end{align*}
Let $\tilde f:\bbR \to \bbC$ be the smooth function
$
\tilde f(s) = f(v_{-s}(i))
$
and note that $\tilde f$ has compact support since $f$ was assumed to vanish in a neighborhood of $\pm 1$.
Hence for $z \in \mathring I_+$ we have
\begin{equation}\label{eqn FN}
F_N(z) =  \frac{N}{\sqrt{2\pi}} \int_{-\infty}^\infty e^{-\tfrac{(N(s+\alpha(z)))^2}{2}} \, \tilde f(s)\,ds = \frac{N}{\sqrt{2\pi}} \int_{-M}^M e^{-\tfrac{(N(s+\alpha(z)))^2}{2}} \, \tilde f(s)\,ds,
\end{equation}
for some $M$ sufficiently large that the interval $[-M,M]$ contains the support of $\tilde f$.
The formula \eqref{eqn FN} now visibly defines a holomorphic function of $z$ in a neighborhood of $\bbD \setminus \{\pm 1\}$ which extends the original definition of $F_N$ on $\mathring I_+$, and we denote this extension again by $F_N$.
This completes Step 1.

In Step 2, we will show that $F_N$ extends further to a smooth function on the closed disk $\bbD$.
We do this by showing that the holomorphic function $F_N|_{\mathring{\bbD}}$ and all of its derivatives extend continuously to $\partial \bbD = S^1$.
From Step 1 we know that there exists such an extension to a neighborhood of $\bbD \setminus \{\pm 1\}$, so it remains to consider the behavior near $z=\pm 1$. We first argue that the integrand of \eqref{eqn FN} is bounded  as a function of  $(s,z) \in [-M,M] \times (\bbD \setminus \{\pm 1\})$.
By inspection, it suffices to show that the function $\operatorname{Re}\big[ (s+\alpha(z))^2 \big]$ is bounded below on the same domain.
And indeed,
\[
\operatorname{Re}\big[ (s+\alpha(z))^2 \big] = -\operatorname{Im}[\alpha(z)]^2 + (s+\operatorname{Re}[\alpha(z)])^2
\]
is visibly bounded below, since $\alpha$ takes values in the strip $\bbS$.

Since the integrand of \eqref{eqn FN} is bounded, we may apply the dominated convergence theorem to obtain
\[
\lim_{z \to \pm 1} F_N(z) = 
\frac{N}{\sqrt{2\pi}} \int_{-M}^M \lim_{z \to \pm 1} e^{-\tfrac{(N(s+\alpha(z)))^2}{2}} \, \tilde f(s)\,ds =
0,
\]
with the limit taken over $z \in \bbD \setminus \{\pm 1\}$.
It follows that $F_N$ extends to a continuous function on all of $\bbD$, vanishing at $\pm 1$.

We next wish to show that the complex derivatives $F_N^{(k)}$, a priori defined on $\mathring{\bbD}$, also extend to continuous functions on $\bbD$ vanishing at $\pm 1$.
We use a similar argument.
Since the integrand of \eqref{eqn FN} is smooth on $[-M,M] \times (\bbD \setminus \{\pm 1\})$, we may differentiate under the integral and see that the iterated $z$-derivatives of $F_N$ are of the form
\begin{equation} \label{eqn: FNk}
F^{(k)}_N(z) =\frac{N}{\sqrt{2\pi}} \int_{-M}^M r(z) e^{-\tfrac{(N(s+\alpha(z)))^2}{2}} \, \tilde f(s)\,ds,
\end{equation}
for holomorphic functions $r(z)$ that are polynomials in $z$, $(z\pm1)^{-1}$ and $\alpha(z)$.

It follows that $F_N^{(k)}$ extends to a continuous function on $\bbD \setminus \{\pm 1\}$, and it remains to analyze the points $z=\pm 1$. Without loss of generality we only consider $z=1$.

Near $z=1$, straightforward computation yields that when $-M \le s \le M$ we have
\[
e^{-\tfrac{(N(s+\alpha(z)))^2}{2}} \le Ce^{-k \Re[\alpha(z)]^2} \le C e^{-k(\log\abs{z-1})^2}.
\]
As $\lim_{t \to 0^+} t^{-k} e^{-(\log t)^2} = \lim_{t \to 0^+} t^{-k-\log t} = 0$ for any $k > 0$,
we see that
\[
\lim_{z\to 1} r(z)e^{-\tfrac{(N(s+\alpha(z)))^2}{2}} = 0,
\]
with the limit taken over $z \in \bbD \setminus \{\pm 1\}$.
An application of the dominated convergence theorem to \eqref{eqn: FNk} now yields that 
\[
\lim_{z \to 1} F_N^{(k)}(z) = 0,
\]
with the limit again taken over $z \in \bbD \setminus \{\pm 1\}$.
The same holds at $z=-1$, and we conclude that each $F_N^{(k)}$ extends to a continuous function on the closed disk $\bbD$, vanishing at $\pm 1$.
It follows that $F_N$ extends to a smooth function on $\bbD$ with the required behavior, completing the proof.

\end{proof}

We will now study the action of analytically continued operators $\widetilde V_\tau$ on $\cP_\cX(I_+)\Omega$, which is a dense subspace $\cD$ by Lemma~\ref{lem: X dense} and the Reeh-Schlieder property \cite[App. A]{CarpiRaymondTanimotoTener25}.
\begin{defn}
If $F \in \cX$ and $\tau \in \bbS := \{0 \le \Im(\tau) \le \pi\}$, we write $\beta_d(v_\tau)F|_{I_+}$ for the function
\begin{align}\label{eqn: beta d F on I}
(\beta_d(v_\tau)F)(z) &= X_{v_\tau}(v_{-\tau}(z))^{d-1} F(v_{-\tau}(z)) \\
&= (\cosh(\tau) + \Re(z)\sinh(\tau))^{d-1}F(v_{-\tau}(z)), & z \in I_+. \nonumber
\end{align}
\end{defn}
The requirement that $\tau \in \bbS$ ensures that $v_{-\tau}(I_+) \subset \bbO$, which is the domain of definition of $F$.
If $f = F|_{I_+} \in C^\infty(S^1)$ and $t \in \bbR$ then $\beta_d(v_t)F|_{I_+} = \beta_d(v_t)f$, which is to say that the above notation is consistent with our usual definition of $\beta_d(v_t)$.
Since $F$ vanishes to all orders at $\pm 1$ we have $\beta_d(v_\tau)F|_{I_+} \in C_0^\infty(I_+)$, which as previously mentioned we typically extend to a function in $C^\infty_0(S^1)$ by extending to be the zero function on $I_-$.
It follows from the definition that
\begin{equation}\label{eqn: compatibility beta d tau and t}
\beta_d(v_t)\big[\beta_d(v_\tau)F|_{I_+}\big] = \beta_d(v_{\tau+t})F|_{I_+},
\end{equation}
giving an analytic continuation of the more standard identity that holds when $\tau \in \bbR$.

\begin{lem}\label{lem: vector valued holo map on strip}
Let $\bbS = \{ \tau \in \bbC \, : \, 0 \le \Im(\tau) \le \pi\}$.
\begin{enumerate}[(i)]
\item Let $F \in \cX$. Then the map $\bbS \to C_0^\infty(I_+)$ given by $\tau \mapsto \beta_d(v_\tau)F|_{I_+}$ is continuous and holomorphic in the interior.

\item Let $F_1, \ldots, F_k \in \cX$ and let $\varphi_1, \ldots \varphi_k \in \cF$.
Then the map $\bbS \to \cD$ given by
\[
\tau \mapsto \varphi_1(\beta_{d_1}(v_\tau)F_1|_{I_+}) \cdots \varphi_k(\beta_{d_k}(v_\tau)F_k|_{I_+})\Omega
\]
is continuous and holomorphic in the interior.
\end{enumerate}
\end{lem}
\begin{proof}
From inspection of the formula \eqref{eqn: beta d F on I} we can see that the map $\bbS \times I_+ \to \bbC$ given by $(\tau,z) \mapsto (\beta_d(v_\tau)F)(z)$ is smooth.
Hence the associated map $\bbS \to C^\infty_0(I_+)$ is smooth, and again from the formula \eqref{eqn: beta d F on I} we see that the derivative $\partial_{\overline{\tau}}$ vanishes on $\mathring{\bbS}$, establishing claim (i).

By the definition of the topology on $\cD$, expressions $\varphi_1(f_1) \cdots \varphi_k(f_k)\Omega$ are jointly continuous in the $f_j$, and combined with part (i) we have that 
\[
\varphi_1(\beta_{d_1}(v_\tau)F_1|_{I_+}) \cdots \varphi_k(\beta_{d_k}(v_\tau)F_k|_{I_+})\Omega
\]
depends continuously on $\tau$.
The usual proof of the product rule shows that this dependence is holomorphic on $\mathring{\bbS}$.
\end{proof}

As a straightforward consequence of Lemma~\ref{lem: vector valued holo map on strip} we have the following, which implies that for all $\tau \in \bbS$ we have that $D(\widetilde V_\tau)$ contains $\cP_{\cX}(I_+)\Omega$, and in particular is densely defined (by Lemma~\ref{lem: X dense}).
The proof of the following lemma is adapted from \cite[p.42-43]{Henriques2014CFTNotes}.
\begin{lem}\label{lem: Vtau action}
Let $\varphi_1,\ldots, \varphi_k \in \cF$, let $F_1, \ldots, F_k \in \cX$, and let $\bbS = \{ 0 \le \Im(\tau) \le \pi\}$.
Let 
$
\Phi = \varphi_1(F_1|_{I_+}) \cdots \varphi_k(F_k|_{I_+})\Omega.
$
Then we have the following.
\begin{enumerate}[(i)]
\item For any $\tau \in \bbS$ we have $\Phi \in D(\widetilde V_\tau)$ and
\begin{equation}\label{eqn: Vtau action}
\widetilde V_\tau\Phi = \varphi_1(\beta_{d_1}(v_\tau)F_1|_{I_+}) \cdots \varphi_k(\beta_{d_k}(v_\tau)F_k|_{I_+})\Omega.
\end{equation}
\item If $d=\sum d_j$ then
\begin{equation}\label{eqn: Vipi action}
\widetilde V_{i\pi}\Phi =  (-1)^d \varphi_k(F_k|_{I_+} \circ z^{-1}) \cdots \varphi_1(F_1|_{I_+} \circ z^{-1})\Omega. 
\end{equation}
\end{enumerate}
\end{lem}
\begin{proof}
Fix $\tau_0 \in \bbS$ and $\lambda \in \cD_\cF^*$.
By Lemma~\ref{lem: vector valued holo map on strip}, the map 
\[
G_{\lambda}(\tau) := \lambda\big(\varphi_1(\beta_{d_1}(v_\tau)F_1|_{I_+}) \cdots \varphi_k(\beta_{d_k}(v_\tau)F_k|_{I_+})\Omega\big)
\]
is continuous on the closed strip bounded by $\bbR$ and $\tau_0 + \bbR$ and holomorphic on the interior, with
\[
G_{\lambda}(t) = \lambda\big(V_t \varphi_1(F_1|_{I_+}) \cdots \varphi_k(F_k|_{I_+})\Omega\big), \qquad t \in \bbR.
\]
In light of \eqref{eqn: compatibility beta d tau and t}, for $t \in \bbR$ we have
\[
G_{\lambda}(t+\tau_0) = \lambda\big(V_t \varphi_1(\beta_{d_1}(v_{\tau_0})F_1|_{I_+}) \cdots \varphi_k(\beta_{d_k}(v_{\tau_0})F_k|_{I_+})\Omega\big), \qquad t \in \bbR.
\]
Hence by Definition~\ref{def: Vs}, $\varphi_1(F_1|_{I_+}) \cdots \varphi_k(F_k|_{I_+})\Omega \in D(\widetilde V_{\tau_0})$ and 
\[
\widetilde V_{\tau_0}\big[\varphi_1(F_1|_{I_+}) \cdots \varphi_k(F_k|_{I_+})\Omega\big] = \varphi_1(\beta_{d_1}(v_{\tau_0})F_1|_{I_+}) \cdots \varphi_k(\beta_{d_k}(v_{\tau_0})F_k|_{I_+})\Omega.
\]
This completes the first part of the proof.

It remains to show that the action of $\widetilde V_{i \pi}$ agrees with \eqref{eqn: Vipi action}.
To that end, note that $v_{i\pi}(z) = z^{-1}$, and so by definition  for $F \in \cX$ we have 
\[
\beta_{d_j}(v_{i\pi})F|_{I_+} = (-1)^{d_j-1} \cdot (F \circ z^{-1})|_{I_+} = (-1)^{d_j-1} \cdot F|_{I_-} \circ z^{-1}
\]
Hence from the formula for $\widetilde V_{\tau}$ above we have
\begin{equation}\label{eqn: vipi intermediate step}
\widetilde V_{i\pi} \Phi = (-1)^{d-k} \varphi_1(F_1 \circ z^{-1}|_{I_+}) \cdots \varphi_k(F_k \circ z^{-1}|_{I_+})\Omega.
\end{equation}

We next claim that if $F \in \cX$ then we have
\begin{equation}\label{eqn: claim value on vacuum}
\varphi(F \circ z^{-1}|_{I_+})\Omega = - \varphi(F \circ z^{-1}|_{I_-})\Omega.
\end{equation}
Indeed $F \circ z^{-1}$ is holomorphic on the disk $\mathring{\bbD}$ and vanishes at the origin, and hence its power series expansion at the origin contains only (strictly) positive powers of $z$. 
It follows that $\varphi(F \circ z^{-1}|_{S^1})$ strictly lowers conformal dimension.
In particular,
\[
\varphi(F \circ z^{-1}|_{S^1})\Omega = 0.
\]
Since $\varphi(F \circ z^{-1}|_{S^1}) = \varphi(F \circ z^{-1}|_{I_+}) + \varphi(F \circ z^{-1}|_{I_-})$, the claim \eqref{eqn: claim value on vacuum} follows.

Following \cite[p.43]{Henriques2014CFTNotes}, we set $g_j = F_j \circ z^{-1}|_{I_+}$ and $h_j = F_j \circ z^{-1}|_{I_-} = F_j|_{I_+} \circ z^{-1}$ and compute
\begin{align*}
\varphi_1(g_1) \cdots \varphi_k(g_k)\Omega & = 
 - \varphi_1(g_1) \cdots \varphi_{k-1}(g_{k-1})\varphi_k(h_k)\Omega\\
 & = -\varphi_k(h_k)\varphi_1(g_1) \cdots \varphi_{k-1}(g_{k-1})\Omega, 
\end{align*}
where the first equality is \eqref{eqn: claim value on vacuum} and the second equality is due to locality.
Repeating this argument yields
\[
\varphi_1(g_1) \cdots \varphi_k(g_k)\Omega 
= (-1)^k\varphi_k(h_k) \cdots \varphi_1(h_1)\Omega.
\]
Combining this identity with \eqref{eqn: vipi intermediate step} yields the desired expression \eqref{eqn: Vipi action}.
\end{proof}

We have thus shown that for $\tau \in \bbS$ the domain $D(\widetilde V_\tau)$ contains $\cP_\cX(I_+)\Omega$, where we recall $\cP_\cX(I_+) \subset \cP(I_+)$ is generated by operators $\varphi(F|_{I_+})$ with $F \in \cX$.
We would like to extend this result to $\cP(I_+)\Omega$ via a limiting argument, approximating $f \in C^\infty_0(I_+)$ by functions of the form $F|_{I_+}$.
We will begin by examining the growth rate of the function $\lambda(V_t\Phi)$.

\begin{lem}\label{lem: growth estimate on R}
Let $\varphi_1, \ldots, \varphi_k \in \cF$ and let $\lambda \in \cD_\cF^*$.
Then there exists a positive integer $N$ and function $C:\bbR \to \bbR_{> 0}$ such that:
\begin{enumerate}[(i)]
\item  there exist constants $k_1,k_2>0$ such that $C(t) \le k_1 e^{k_2\abs{t}}$ for all $t \in \bbR$
\item for all $f_j,g_j \in C^\infty(S^1)$ and all $t \in \bbR$ we have
\begin{align*}
&\abs{\lambda\big(\big[V_t\varphi_1(f_1) \cdots \varphi_k(f_k) \Omega\big]-\big[V_t\varphi_1(g_1) \cdots \varphi_k(g_k) \Omega\big]\big)} \le\\
&\qquad C(t)\left[\prod_{j=1}^k \big(1+\norm{f_j}_{C^N}+\norm{g_j}_{C^N}\big) \right]\max_{1 \le j \le k} \norm{f_j-g_j}_{C^N},
\end{align*}
where $\norm{f}_{C^N} = \sum_{j=0}^N \big\|\tfrac{d^j}{d\theta^j} f\big\|_\infty$.
\end{enumerate}
\end{lem}
\begin{proof}

The map
\[
(h_1, \ldots, h_k) \mapsto \lambda\big(\varphi_1(h_1) \cdots \varphi_k(h_k)\Omega\big)
\]
is jointly continuous on $C^\infty(S^1)^k$.
Thus by standard results (cf. the beginning of \cite[\S41]{Treves67}) there exists a positive integer $N$ such that
\[
\abs{\lambda\big(\varphi_1(h_1) \cdots \varphi_k(h_k)\Omega\big)} \le \norm{h_1}_{C^N} \cdots \norm{h_k}_{C^N}, \qquad h_j \in C^\infty(S^1),
\]
after potentially rescaling $\lambda$ for notational convenience.
Observe that
\begin{align*}
\varphi_1(f_1) \cdots \varphi_k(f_k) - \varphi_1(g_1) \cdots \varphi_k(g_k) = 
\sum_{j=1}^k \left[ \prod_{\ell=1}^{j-1} \varphi_\ell(f_\ell) \right ]\varphi_j(f_j-g_j) \left[ \prod_{\ell=j+1}^{k} \varphi_\ell(g_\ell) \right ],
\end{align*}
where we take the convention that the (non-commuting) products are ordered left-to-right, and the empty products (when $j=1$ or $j=k$) are omitted.
Hence
\begin{align*}
&\abs{\lambda\big(\big[V_t\varphi_1(f_1) \cdots \varphi_k(f_k) \Omega\big]-\big[V_t\varphi_1(g_1) \cdots \varphi_k(g_k) \Omega\big]\big)} \le\\
&\qquad \le \sum_{j=1}^k \abs{\lambda\left( \left[ \prod_{\ell=1}^{j-1} \varphi_\ell(\beta_{d_\ell}(v_t) f_\ell) \right ]\varphi_j(\beta_{d_j}(v_t)[f_j-g_j]) \left[ \prod_{\ell=j+1}^{k} \varphi_\ell(\beta_{d_\ell}(v_t)g_\ell) \right]\Omega\right)} \le\\
&\qquad \le k \left[\prod_{j=1}^k \big(1+\norm{\beta_{d_j}(v_t)f_j}_{C^N}+\norm{\beta_{d_j}(v_t)g_j}_{C^N}\big) \right]\max_{1 \le j \le k} \norm{\beta_{d_j}(v_t)[f_j-g_j]}_{C^N}.
\end{align*}
Thus it suffices to show that for any $f \in C^\infty(S^1)$ the function $\norm{\beta_d(v_t)f}_{C^N}$ satisfies
\begin{equation}\label{eqn: exp growth bound CN norm}
\norm{\beta_d(v_t)f}_{C^N} \le C(t) \norm{f}_{C^N}
\end{equation}
for a function $C:\bbR \to \bbR_{> 0}$ satisfying the growth condition (i).
We say that a function $\bbR \to \bbR_{> 0}$ satisfying (i) has `at most exponential growth,' and we note that the class of such functions is closed under addition and products.

We first consider $\norm{\beta_d(v_t)f}_\infty$, i.e. the case $N=0$.
Recall that
\[
(\beta_d(v_t)f)(z) = (\cosh(t) + \Re(z) \sinh(t))^{d-1} \cdot f(v_{-t}(z)).
\]
Hence
\[
\norm{\beta_d(v_t)f}_\infty \le \left( \sup_{z \in S^1} \abs{\cosh(t) + \Re(z) \sinh(t)} \right)^{d-1} \norm{f}_\infty
\]
and the function $t \mapsto \sup_{z \in S^1} \abs{\cosh(t) + \Re(z) \sinh(t)}$ evidently has at most exponential growth (it is, in fact, exactly $e^{\abs{t}}$).
This establishes the necessary estimate in the case $N=0$.

More generally, the derivative $\tfrac{d^j}{d\theta^j} \beta_d(v_t)f$ is a linear combination of products of functions of the form 
\[
\tfrac{d^k}{d\theta^k}\big[\cosh(t) + \Re(z) \sinh(t)\big], \quad \tfrac{d^k}{d\theta^k} v_{-t}(z), \quad f^{(k)}(v_{-t}(z)), \qquad 0 \le k \le j,
\]
with exactly one term of the last type in each product.
Explicit calculation shows that the sup norms of functions of the first two types are bounded by functions of $t$ with at most exponential growth as above, and we conclude that \eqref{eqn: exp growth bound CN norm} holds for all $N$ with a function $C$ with at most exponential growth.
\end{proof}

The next lemma establishes an analogous growth bound on the analytic continuation $\lambda(\widetilde V_\tau\Phi)$ on the strip $\bbS$ when $\Phi \in \cP_{\cX}(I_+)\Omega$.
Of particular note is that the estimate only involves data related to the vector $\Phi$ (and the function $V_t\Phi$ on one boundary of the strip), and does not reference $\widetilde V_{i\pi}\Phi$.

\begin{lem}\label{lem: growth estimate on S}
Let $\varphi_1, \ldots, \varphi_k \in \cF$ and let $\lambda \in \cD_\cF^*$.
Then there exists a positive integer $N$ and positive constant $M$ with the following property.
Let $F_j,G_j \in \cX$ and let $f_j = F_j|_{I_+}$ and $g_j = G_j|_{I_+}$.
Then for all $\tau \in \bbS = \{0 \le \Im(\tau) \le \pi\}$ we have
\begin{align*}
&\abs{\lambda\big(\big[\widetilde V_\tau\varphi_1(f_1) \cdots \varphi_k(f_k) \Omega\big]-\big[\widetilde V_\tau\varphi_1(g_1) \cdots \varphi_k(g_k) \Omega\big]\big)} \le\\
&\qquad M e^{\Re(\tau)^2} \left[\prod_{j=1}^k \big(1+\norm{f_j}_{C^N}+\norm{g_j}_{C^N}\big) \right]\max_{1 \le j \le k} \norm{f_j-g_j}_{C^N}.
\end{align*}
\end{lem}
\begin{proof}
Let us temporarily fix $F_j,G_j \in \cX$ along with the corresponding $f_j=F_j|_{I_+}$ and $g_j=G_j|_{I_+}$. 
Let $\mathfrak{F},\mathfrak{G}:\bbS \to \bbC$ be the functions
\[
\mathfrak{F}(\tau) = \lambda\big(\widetilde V_\tau\varphi_1(f_1) \cdots \varphi_k(f_k) \Omega\big), \quad \mathfrak{G}(\tau) = \lambda\big(\widetilde V_\tau\varphi_1(g_1) \cdots \varphi_k(g_k) \Omega\big)
\]
and let $\mathfrak{H} = \mathfrak{F}-\mathfrak{G}$.
By Lemma~\ref{lem: vector valued holo map on strip}, the functions $\mathfrak{F},\mathfrak{G}$ and $\mathfrak{H}$ are continuous on the closed strip $\bbS = \{0 \le \Im(\tau) \le \pi\}$ and holomorphic on the interior.
Write $\tau = t + i s$ with $0 \le s \le \pi$, and let 
\begin{equation}\label{eqn: fjs}
f_{j,s} = \beta_{d_j}(v_{is})F_j|_{I_+}, \qquad g_{j,s} = \beta_{d_j}(v_{is})G_j|_{I_+},
\end{equation}
so that 
\[
\widetilde V_{is} \varphi_1(f_1)\cdots\varphi_k(f_k)\Omega = \varphi_1(f_{1,s}) \cdots \varphi_k(f_{k,s})\Omega
\]
by Lemma~\ref{lem: Vtau action}(i), and similarly for the $g_j$.
By Lemma~\ref{lem: growth estimate on R} applied to the functions $f_{j,s}$ and $g_{j,s}$ (along with the fact that $\widetilde V_{t+is} = V_t \widetilde V_{is}$) we have 
\begin{align*}
\abs{\mathfrak{H}(t+is)} \le C(t)\left[\prod_{j=1}^k \big(1+\norm{f_{j,s}}_{C^N}+\norm{g_{j,s}}_{C^N}\big) \right]\max_{1 \le j \le k} \norm{f_{j,s}-g_{j,s}}_{C^N}
\end{align*}
for some positive integer $N$ and some function $C$ with at most exponential growth that does not depend on the functions $F_j$ and $G_j$.
Since $F_j$ and $G_j$ are smooth on $\bbO$, the maximum
\[
\max_{0 \le s \le \pi} \left[\prod_{j=1}^k \big(1+\norm{f_{j,s}}_{C^N}+\norm{g_{j,s}}_{C^N}\big) \right]\max_{1 \le j \le k} \norm{f_{j,s}-g_{j,s}}_{C^N}
\]
is achieved at some $s_0 \in [0,\pi]$.
Hence there is a constant $K > 0$ (depending on the $F_j$ and $G_j$) such that
\[
\abs{\mathfrak{H}(\tau)} \le K \, C(\Re \tau).
\]
Since $\big|e^{-\tau^2}\big| \le e^{\pi^2} e^{-\Re(\tau)^2}$ on $\bbS$ and $C$ has exponential growth, we see that $e^{-\tau^2} \mathfrak{H}(\tau)$ is bounded, and in fact vanishes at $\infty$.
A straightforward application of the (strong) maximum principle yields that $\big|e^{-\tau^2}\mathfrak{H}(\tau) \big|$ achieves its maximum on $\partial \bbS = \bbR \cup \bbR + i\pi$.

We first investigate the dependence of this maximum on the functions $F_j$ and $G_j$, although we still suppress this dependence from the notation $\mathfrak{H}(\tau)$.
For $t \in \bbR$ we have
\begin{equation}\label{eqn: H bound on R}
\big| e^{-t^2}\mathfrak{H}(t)\big| \le M \left[\prod_{j=1}^k \big(1+\norm{f_{j}}_{C^N}+\norm{g_{j}}_{C^N}\big) \right]\max_{1 \le j \le k} \norm{f_{j}-g_{j}}_{C^N}
\end{equation}
where $M = \max_t e^{-t^2} C(t)$ does not depend on the functions $F_j$,$G_j$.

We now consider the behavior of $e^{-\tau^2} \mathfrak{H}(\tau)$ on $\bbR + i \pi$.
Recall from Lemma~\ref{lem: Vtau action} that
\[
\widetilde V_{i\pi}\varphi_1(f_1) \cdots \varphi_k(f_k)\Omega = (-1)^d \varphi_k(\hat f_k) \cdots \varphi_1(\hat f_1)\Omega
\]
where $\hat f_j = f_j \circ z^{-1}$ and $d = \sum d_j$.
Crucially, $\|\hat f_j\|_{C^N} = \norm{f_j}_{C^N}$.
We have similar computations for the functions $g_j$, so by Lemma~\ref{lem: growth estimate on R}
\begin{align*}
\abs{\mathfrak{H}(t+i\pi)} &= \abs{\lambda \big( \big[ \widetilde V_{t+i\pi}\varphi_1(f_1) \cdots \varphi_k(f_k)\Omega\big] - \big[\widetilde V_{t+i\pi}\varphi_1(g_1) \cdots \varphi_k(g_k)\Omega\big]\big)}\\
&= \abs{\lambda\big(\big[V_{t}\varphi_k(\hat f_k) \cdots \varphi_1(\hat f_1)\Omega\big] - \big[V_{t}\varphi_k(\hat g_k) \cdots \varphi_1(\hat g_1)\Omega\big]\big)}\\
&\le D(t) \left[\prod_{j=1}^k \big(1+\|\hat f_j\|_{C^{N'}}+\norm{\hat g_j}_{C^{N'}}\big) \right]\max_{1 \le j \le k} \|\hat f_j- \hat g_j\|_{C^{N'}}\\
&= D(t) \left[\prod_{j=1}^k \big(1+\norm{f_j}_{C^{N'}}+\norm{g_j}_{C^{N'}}\big) \right]\max_{1 \le j \le k} \norm{f_j- g_j}_{C^{N'}}.
\end{align*}
The function $D(t)$ and integer $N'$ may differ from the previous computation because of the reordering of the fields $\varphi_j$, but by replacing $N$ or $N'$ with $\max \{N,N'\}$ we may assume without loss of generality that $N=N'$.
Arguing as before we have
\[
\abs{e^{-(t+i\pi)^2} \mathfrak{H}(t+i\pi)} \le M' \left[\prod_{j=1}^k \big(1+\norm{f_{j}}_{C^N}+\norm{g_{j}}_{C^N}\big) \right]\max_{1 \le j \le k} \norm{f_{j}-g_{j}}_{C^N}
\]
where $M' = \max e^{\pi^2} e^{-t^2}D(t)$ does not depend on the functions.
Combined with \eqref{eqn: H bound on R}, we have for all $\tau \in \partial\bbS = \bbR \cup \bbR + i \pi$
\[
\abs{e^{-\tau^2} \mathfrak{H}(\tau)} \le M'' \left[\prod_{j=1}^k \big(1+\norm{f_{j}}_{C^N}+\norm{g_{j}}_{C^N}\big) \right]\max_{1 \le j \le k} \norm{f_{j}-g_{j}}_{C^N}
\]
where $M''=\max\{M,M'\}$.
As discussed previously, $\abs{\mathfrak{H}(\tau)}$ is maximized on $\partial \bbS$, so as required we have for all $\tau \in \bbS$
\begin{align*}
\abs{\mathfrak{H}(\tau)} \le M'' \big| e^{\tau^2}\big| \left[\prod_{j=1}^k \big(1+\norm{f_{j}}_{C^N}+\norm{g_{j}}_{C^N}\big) \right]\max_{1 \le j \le k} \norm{f_{j}-g_{j}}_{C^N}\\
\le M'' e^{\Re(\tau)^2} \left[\prod_{j=1}^k \big(1+\norm{f_{j}}_{C^N}+\norm{g_{j}}_{C^N}\big) \right]\max_{1 \le j \le k} \norm{f_{j}-g_{j}}_{C^N}.
\end{align*}
\end{proof}

We can now give the main result of Section~\ref{sec: general nonunitary}, computing $\widetilde V_{\pm i\pi} x\Omega$ for $x \in \cP(I_\pm)$.

\begin{thm}\label{thm: nonunitary Vipi}
Let $(\cF,\cD,U,\Omega)$ be a non-unitary Wightman CFT on $S^1$.
Let $\cP(I_+) \subset \End(\cD)$ be the algebra generated by operators $\varphi(f)$ with $\varphi \in \cF$ and $\supp(f) \subset I_+$.
Let $\widetilde V_{\pm i\pi}$ be the densely defined operators from Definition~\ref{def: Vs} with domains $D(\widetilde V_{\pm i\pi})$.
\begin{enumerate}[(i)]
    \item $\cP(I_+)\Omega \subset D(\widetilde V_{i\pi})$ and $\cP(I_-)\Omega \subset D(\widetilde V_{-i\pi})$
    \item Let $\varphi_1, \ldots, \varphi_k \in \cF$ with conformal dimensions $d_j$, and let $f_1, \ldots, f_k \in C^\infty(S^1)$ with $\operatorname{supp} f_j \subset I_+$. Let $d= \sum d_j$. Then we have
    \[
    \widetilde V_{i\pi}\varphi_1(f_1)\cdots \varphi_k(f_k)\Omega = (-1)^d \varphi_k(f_k \circ z^{-1}) \cdots \varphi_1(f_1 \circ z^{-1})\Omega.
    \]
    \item The analogous formula holds for $\widetilde V_{-i\pi}$ and $\operatorname{supp} f_j \subset I_-$.
\end{enumerate}
\end{thm}
\begin{proof}
Without loss of generality, we only consider $I_+$.
If for all $j$ we have $f_j = F_j|_{I_+}$ for some $F_j \in \cX$ then by Lemma~\ref{lem: Vtau action} we have
\[
\widetilde V_{i\pi}x\Omega = (-1)^d \varphi_k(F_k|_{I_+} \circ z^{-1}) \cdots \varphi_1(F_1|_{I_+} \circ z^{-1})\Omega
\]
as required.
We now consider general $f_j \in C^\infty(S^1)$ with $\supp f_j \subset I_+$.
By Lemma~\ref{lem: X dense}, for each $j=1,\ldots, k$ there exists a sequence $F_{j,n} \in \cX$ such that $\lim_{n \to \infty} F_{j,n}|_{I_+} = f_j$ in the $C^\infty$ topology.
Let $f_{j,n} = F_{j,n}|_{I_+}$ and fix $\lambda \in \cD_\cF^*$.
Note that $\norm{f_{j,n}}_{C^N}$ is bounded for every positive integer $N$, so by Lemma~\ref{lem: growth estimate on S} we have for all $\tau \in \bbS=\{0 \le \Im(\tau) \le \pi\}$ 
\begin{align*}
\big|\lambda\big(\big[\widetilde V_\tau &\varphi_1(f_{1,n}) \cdots \varphi_k(f_{k,n})\Omega\big] - \big[\widetilde V_\tau \varphi_1(f_{1,m}) \cdots \varphi_k(f_{k,m})\Omega\big]\big)\big| \\
&\le Me^{\Re(\tau)^2} \max_{1 \le j \le k} \norm{f_{j,n}-f_{j,m}}_{C^N}.
\end{align*}
for some $M>0$ that does not depend on $m$ and $n$.
Thus we see that the sequence of holomorphic functions
\[
G_{\lambda,n}(\tau) = \lambda\big(\widetilde V_\tau \varphi_1(f_{1,n}) \cdots \varphi_k(f_{k,n})\Omega\big)
\]
is locally uniformly Cauchy on $\bbS$, and thus converges locally uniformly to a function $G_{\lambda}$.
Since each $G_{\lambda,n}$ is continuous on $\bbS$ and holomorphic on the interior (by Lemma~\ref{lem: vector valued holo map on strip}), so is $G_{\lambda}$.
Taking pointwise limits, we see that for $t \in \bbR$
\[
G_{\lambda}(t) = \lambda\big(V_t \varphi_1(f_1) \cdots \varphi_k(f_k)\Omega\big).
\]
Again invoking the formula for $\widetilde V_{i\pi}$ from Lemma~\ref{lem: Vtau action} and the fact that $\widetilde V_{t+i\pi} = V_t \widetilde V_{i\pi}$, we have
\begin{align*}
G_{\lambda}(t+i\pi) & = \lim_{n \to \infty} (-1)^d \lambda\big(V_t \varphi_k(f_{k,n} \circ z^{-1}) \cdots \varphi_1(f_{1,n} \circ z^{-1})\Omega\big)\\
& = (-1)^d \lambda\big(V_t \varphi_k(f_k \circ z^{-1}) \cdots \varphi_1(f_1 \circ z^{-1})\Omega\big).
\end{align*}
Thus by Definition~\ref{def: Vs} we have $\varphi_1(f_1) \cdots \varphi_k(f_k)\Omega \in D(\widetilde V_{i\pi})$ and 
\[
\widetilde V_{i\pi}\varphi_1(f_1) \cdots \varphi_k(f_k)\Omega = (-1)^d \varphi_k(f_k \circ z^{-1}) \cdots \varphi_1(f_1 \circ z^{-1})\Omega. \qedhere
\] 
\end{proof}

\section{Invariant Hermitian forms, Bisognano-Wichmann, and the KMS condition}\label{sec: sesquilinear}

In this section, we  assume that our Wightman CFT $(\cF,\cD,U,\Omega)$ has an invariant nondegenerate sesquilinear form $\ip{\, \cdot \, , \, \cdot \,}$, along with an antiunitary PCT involution $\theta:\cD \to \cD$ which fixes $\Omega$.
Such a theory is called \textbf{involutive} (see \cite[\S4]{CarpiRaymondTanimotoTener25} for additional background and the connection with vertex algebras).
The invariance of the form $\ip{ \, \cdot \, , \, \cdot \,}$ means, firstly, that
\begin{equation}\label{eqn: invariant sesquilinear}
\ip{\varphi(f)\Phi,\Psi} = \ip{\Phi,\varphi^\dagger(\overline{f})\Psi}, \qquad \Phi,\Psi \in \cD, \, \varphi \in \cF, \, f \in C^\infty(S^1)
\end{equation}
where $\varphi^\dagger$ is the operator-valued distribution 
\[
\varphi^\dagger(f) =(-1)^{d_\varphi} \theta \varphi (\overline{f} \circ z^{-1}) \theta,
\]
and moreover we have $(-1)^{d_\varphi} \varphi^\dagger \in \cF$.
Secondly, an invariant form must also satisfy
\begin{equation}\label{eqn: involutive mobius unitary}
\ip{U(\gamma)\Phi,U(\gamma)\Psi} = \ip{\Phi,\Psi}, \qquad \Phi,\Psi \in \cD, \, \gamma \in \Mob.
\end{equation}
The sesquilinear form is assumed to be jointly $\cF$-strong continuous and is automatically Hermitian symmetric, meaning $\ip{\Phi,\Psi}=\overline{\ip{\Psi,\Phi}}$ for all $\Phi,\Psi \in \cD$.
Given $\theta$, the existence of a sesquilinear form on $\cD$ with these properties is equivalent to $\ip{ \, \cdot \, , \theta \, \cdot \,}$ being a nondegenerate invariant bilinear form on the M\"obius vertex algebra underlying $\cF$.
The theory is called \textbf{unitary} if the Hermitian form is an inner product (normalized so that $\norm{\Omega} = 1$).
In terms of the correspondence between Wightman CFTs and vertex algebras, an involutive Wightman CFT corresponds to a M\"obius vertex algebra with a sesquilinear form and antiunitary operator $\theta$ such that $\ip{\, \cdot \, \theta \, \cdot \,}$ is an invariant bilinear form in the sense of \cite{FHL93,Li94}.
See \cite[Def. 4.9]{CarpiRaymondTanimotoTener25} and the surrounding discussion for more detail.

Recall that $\End(\cD)$ denotes the collection of $\cF$-strong continuous linear operators on $\cD$. 
Given $x \in \End(\cD)$, we say that $x$ has an adjoint if there exists $x^* \in \End(\cD)$ such that $\ip{x\Phi,\Psi}=\ip{\Phi,x^*\Psi}$ for all $\Phi,\Psi \in \cD$\footnote{See also the discussion following Corollary~\ref{cor: KMS} regarding the partially defined adjoint of a partially defined (and not necessarily continuous) linear map.}, and write $\End_*(\cD)$ for the collection of operators $x \in \End(\cD)$ that have an adjoint $x^* \in \End(\cD)$.
In light of \eqref{eqn: invariant sesquilinear} we see that $\varphi(f) \in \End_*(\cD)$, and moreover the algebras $\cP(I)$ are $*$-algebras, i.e., they are closed under taking adjoints.
We have $\theta \cP(I) \theta = \cP(\overline{I})$, where $\overline{I}$ is the complex conjugate interval, and $\theta$ commutes with $V_t$.
If $\gamma \in \Mob$, then \eqref{eqn: involutive mobius unitary} implies that $U(\gamma) \in \End_*(\cD)$, with $U(\gamma)^* = U(\gamma^{-1})$.

Since the Hermitian form is continuous, each map $\ip{ \, \cdot \, , \Phi}$ lies in $\cD_\cF^*$. 
The collection of such functionals separates points in $\cD$ because the form is nondegenerate, although it is generally a proper subspace of $\cD_\cF^*$.
In Definition~\ref{def: Vs}, we defined the domain of $\widetilde V_\tau$ to consist of vectors $\Phi \in \cD$ such that the functions $\lambda(V_t\Phi)$ have appropriate analytic continuations.
In the context of an involutive theory, where we are interested in the adjoints of operators, it is more appropriate to work on an (a priori) larger domain where we only require analytic continuation of the functions $\ip{V_t\Phi,\Phi'}$.

\begin{defn}\label{defn: Vs involutive}
For $\tau \in \bbC$, the operator $V_\tau$ has domain $D(V_\tau) \subset \cD$ consisting of vectors $\Phi \in \cD$ such that for every $\Phi' \in \cD$ the function $\bbR \to \bbC$ given by
\[
t \mapsto \ip{V_t\Phi,\Phi'}
\]
extends to a (necessarily unique) function $G_{\Phi'}$ on the closed strip $\bbS_\tau$ bounded by $\bbR$ and $\bbR + \tau$ such that $G_{\Phi'}$ is continuous on $\bbS_\tau$ and holomorphic in the interior $\mathring{\bbS_\tau}$. Moreover, we require that there exist $\Psi \in \cD$ such that $G_{\Phi'}(\tau+t) = \ip{V_t\Psi,\Phi'}$ for all $t \in \bbR$ and $\Phi' \in \cD$.
The operator $V_\tau$ is defined on $\Phi \in D(V_\tau)$ by $V_\tau\Phi = \Psi$.
\end{defn}

It is immediate that $D(\widetilde V_\tau) \subset D(V_\tau)$ and 
\begin{equation}\label{eqn: Vtau extension}
V_\tau|_{D(\widetilde V_\tau)} = \widetilde V_\tau.
\end{equation}
As a corollary of Theorem~\ref{thm: nonunitary Vipi} we obtain the Bisognano-Wichmann property \cite[Thm. 1(d)]{BisognanoWichmann75} for involutive Wightman CFTs on $S^1$.

\begin{thm}[Bisognano-Wichmann property]\label{thm: involutive BW}
Let $(\cF,\cD,U,\Omega)$ be an involutive Wightman CFT on $S^1$ with PCT involution $\theta$.
Let $\cP(I_+) \subset \End_*(\cD)$ be the unital algebra generated by operators $\varphi(f)$ with $\varphi \in \cF$ and $\supp f  \subset I_+$.
Let $V_{\pm i\pi}$ be the densely defined operators from Definition~\ref{defn: Vs involutive} with domains $D(V_{\pm i\pi})$.
\begin{enumerate}[(i)]
    \item We have $\cP(I_+)\Omega \subset D(V_{i\pi})$  and if $x \in \cP(I_+)$ then
    \[V_{i\pi}x\Omega = \theta x^* \Omega.\]
    \item We have $\cP(I_-)\Omega \subset D(V_{-i\pi})$  and if $x \in \cP(I_-)$ then
    \[V_{-i\pi}x\Omega = \theta x^* \Omega.\]
\end{enumerate}
\end{thm}
\begin{proof}
We only consider the first item and the second is similar.
By Theorem~\ref{thm: nonunitary Vipi} and \eqref{eqn: Vtau extension} we have $\cP(I_+)\Omega \subset D(V_{i\pi})$.
Consider $x = \varphi_1(f_1) \cdots \varphi_k(f_k)$, where $\varphi_j \in \cF$ with conformal dimension $d_j$ and $\operatorname{supp} f_j \subset I_+$.
 By the same Theorem, we have
\[
V_{i\pi}x\Omega =V_{i \pi} \varphi_1(f_1) \cdots \varphi_k(f_k)\Omega = (-1)^d \varphi_k(f_k \circ z^{-1}) \cdots \varphi_1(f_1 \circ z^{-1})\Omega
\]
where $d = \sum d_j$.
On the other hand, $\theta \Omega = \Omega$ and by \eqref{eqn: invariant sesquilinear} we have
\[
\theta \varphi_j(f_j) \theta = (-1)^{d_j} \varphi_j(f_j \circ z^{-1})^*.
\]
Hence 
\[
\theta x^*\Omega = (-1)^d \varphi_k(f_k \circ z^{-1}) \cdots \varphi_1(f_1 \circ z^{-1})\Omega = V_{i\pi} x \Omega
\]
as required.
\end{proof}

\begin{cor}[KMS condition] \label{cor: KMS}
Let $\phi:\cP(I_+) \to \bbC$ be the linear functional 
\[
\phi(x)=\ip{x\Omega,\Omega}
\]
and let $\alpha_t = \operatorname{Ad} V_t \in \operatorname{Aut}(\cP(I_+))$.
Let $\bbS_{2\pi i} = \{\tau \in \bbC \, : \, 0 \le \Im(\tau) \le 2\pi\}$.
Then for every $x,y \in \cP(I_+)$, there exists a function $F:\bbS_{2\pi i} \to \bbC$ that is continuous on $\bbS_{2\pi i}$, holomorphic on $\mathring{\bbS}_{2 \pi i}$, and satisfies for all $t \in \bbR$
\[
F(t) = \phi(x\alpha_t(y)), \qquad F(t+2\pi i) = \phi(\alpha_t(y)x).
\]
\end{cor}
\begin{proof}
    Let $x,y \in \cP(I_+)$.
    By Theorem~\ref{thm: involutive BW}, there is a function $F_1:\bbS_{i\pi} \to \bbC$ that is continuous and holomorphic on the interior such that
    \[
    F_1(t) = \ip{V_t y\Omega, x^*\Omega}, \qquad F_1(t+i\pi) = \ip{V_tV_{i\pi}y\Omega, x^*\Omega} = \ip{V_t\theta y^*\Omega, x^*\Omega}.
    \]
    Similarly, since $\theta x^* \theta \in \cP(I_-)$, there exists a function $F_2:\bbS_{-i\pi} \to \bbC$ such that
    \[
F_2(t) = \ip{V_t \theta x^* \Omega, y^*\Omega}, \qquad 
    F_2(t-i\pi) = \ip{V_t x\Omega, y^*\Omega}.
\]
Let $G:\{\pi \le \Im(\tau) \le 2\pi\} \to \bbC$ be the function $G(\tau) = F_2(-\tau+i\pi)$. 
We can compute
\[
G(t+i\pi) = F_2(-t) = \ip{V_{-t}\theta x^*\Omega, y^*\Omega} = \ip{\theta V_t y^*\Omega, x^*\Omega} = F_1(t+i\pi)
\]
for all $t \in \bbR$, where we have used the fact that $V_t^* = V_{-t}$ in the third equality and the fact that $\theta$ and $V_t$ commute in the final equality.
Hence $F_1$ and $G$ assemble to a function $F$ continuous on $\bbS_{2\pi i}$ and holomorphic in the interior such that
\[
F(t) = F_1(t) = \phi(x\alpha_t(y)), \qquad F(t + 2\pi i) = F_2(-t-\pi i) = \phi(\alpha_t(y)x),
\]
as required.
\end{proof}

We briefly digress on the subject of partially defined operators in the present context.
For the most part, the definitions and results are identical to the situation of unbounded operators on a Hilbert space.
We give a brief summary of the aspects of the theory which persist into the present context.

A partially defined operator $X$ on $\cD$ consists of a subspace $D(X) \subset \cD$ and a linear map $X:D(X) \to \cD$.
We typically require that $D(X)$ be dense.
We say that $X_1$ is an extension of $X_2$ if $D(X_2) \subseteq D(X_1)$ and $X_1|_{D(X_2)} = X_2$, in which case we write $X_2 \subset X_1$.
If $X$ is densely defined, then the adjoint $X^*$ has domain consisting of vectors $\Psi \in \cD$ such that there exists $\Psi' \in \cD$ for which 
\[
\ip{X\Phi, \Psi} = \ip{\Phi,\Psi'}, \qquad \Phi \in D(X).
\]
In this case $X^*\Psi = \Psi'$.
If $X_2 \subset X_1$ then $X_1^* \subset X_2^*$.
A densely defined operator is called symmetric if $X \subset X^*$, and self-adjoint if $X=X^*$ as partially defined operators (i.e. including having the same domain).
A densely defined operator is called closed if its graph is a closed subset of $\cD \times \cD$, and called closable if the closure of its graph is the graph of a (closed) linear map.
We write $\overline{X}$ for the closure of a closable operator $X$.
A core for a closed operator $X$ is a subspace $W \subset D(X)$ such that $\overline{X|_W} = X$.
The adjoint of a densely defined operator is necessarily closed, and if $X$ has a densely defined adjoint, then it is closable.

It is at this point that we encounter a key difference from the theory of unbounded operators on Hilbert spaces.
If $X$ is closable, then $\overline{X} \subseteq X^{**}$, but equality may fail.
In order to understand this phenomenon, one should recall that the graph of $X^*$ is essentially the orthogonal complement of the graph of $X$.
Thus this is a manifestation of another fundamental departure, which is that if $W \subset \cD$ then $\overline{W} \subset W^{\perp\perp}$, but again equality may fail.
This is a consequence of the fact that $W^\perp$ may be `surprisingly small,' even when the Hermitian form is positive definite\footnote{Consider the case where $\cD$ is the subspace of $L^2[0,1]$ spanned by $L^2[0,\tfrac12]$ and polynomials. The orthogonal complement of $W=L^2[0,\tfrac12]$ in $\cD$ consists of only the zero function.}.

The challenges described above can make it difficult to study duality phenomena in the present context, which in turn increases the difficulty in precisely identifying the domains of various operators.
For example, a question that we do not resolve in this article is the following.

\begin{ques}\label{ques: core}
    Is $\cP(I_+)\Omega$ a core for $V_{i\pi}$?
\end{ques}

In light of the above discussion, it is particularly satisfying that we will now be able to establish (`by hand') that the operators $V_{\pm i\pi}$ are self-adjoint on their given domain.
We first give the following straightforward lemma, which shows that they are symmetric.
\begin{lem}\label{lem: formal adjoint of Vtau}
If $\Phi \in D(V_\tau)$ and $\Phi' \in D(V_{-\overline{\tau}})$, then
\[
\ip{V_\tau\Phi,\Phi'} = \ip{\Phi,V_{-\overline{\tau}}\Phi'}.
\]
\end{lem}
\begin{proof}
Fix $\tau_0 \in \bbC$.
For notational clarity, we assume without loss of generality that $\Im(\tau_0) \ge 0$.
Fix $\Phi \in D(V_{\tau_0})$ and $\Phi' \in D(V_{-\overline{\tau_0}})$.
As $\Phi \in D(V_\tau)$, there is a map $G$ continuous on the strip $\{0 \le \Im(\tau) \le \Im(\tau_0)\}$ and holomorphic on the interior, satisfying
\[
G(t) = \ip{V_t\Phi,\Phi'}, \qquad G(t+\tau_0) = \ip{V_t V_{\tau_0}\Phi,\Phi'} = \ip{V_{\tau_0}\Phi,V_{-t}\Phi'},
\]
where the last equality follows from \eqref{eqn: involutive mobius unitary}.
Similarly, since $\Phi' \in D(V_{-\overline{\tau_0}})$ we have a map $H$ on the same strip (continuous on the closure and holomorphic on the interior) such that
\[
H(t) = \ip{V_t\Phi',\Phi}, \qquad H(t-\overline{\tau_0}) = \ip{V_{-\overline{\tau_0}}\Phi',V_{-t}\Phi}.
\]
By inspection, $\tau \mapsto \overline{H(-\overline{\tau})}$ is again continuous on the strip, holomorphic on the interior, and agrees with $G(\tau)$ on $\bbR$. 
Hence the two functions agree identically, and evaluating at $\tau_0$ yields the desired identity.
\end{proof}

We now give a careful analysis of the domain $D(V_{\pm i\pi}^*)$ in order to show that these operators are self-adjoint.

\begin{thm}\label{thm: Vipi selfadjoint}
Let $(\cF,\cD,U,\Omega)$ be an involutive Wightman CFT on $S^1$, and let
 $V_{\pm i\pi}$ be the densely defined operators from Definition~\ref{defn: Vs involutive} with domains $D(V_{\pm i\pi})$.
Then we have $V_{\pm i\pi}^* = V_{\pm i\pi}$.
Moreover, the adjoints of the restrictions $V_{\pm i\pi}|_{\cP(I_\pm)\Omega}$ are again given by $\big( V_{\pm i\pi}|_{\cP(I_\pm)\Omega}\big)^* = V_{\pm i\pi}$, where $\cP(I_{\pm})$ are the unital algebras generated by operators $\varphi(f)$ with $\varphi \in \cF$ and $\supp f  \subset I_\pm$.

\end{thm}
\begin{proof}
We only consider $V_{i\pi}$.
Observe that we have a chain of extensions of densely defined operators
\[
V_{i\pi}|_{\cP(I_+)\Omega} \subseteq V_{i\pi} \subseteq V_{i\pi}^* \subseteq (V_{i\pi}|_{\cP(I_+)\Omega})^*,
\]
where the first inclusion is immediate, the second inclusion is Lemma~\ref{lem: formal adjoint of Vtau}, and the third follows by taking adjoints in the first inclusion.
To prove the theorem, it therefore suffices to show that $(V_{i\pi}|_{\cP(I_+)\Omega})^* \subset V_{i\pi}$.
In fact, since we know that these two operators agree on $D(V_{i\pi})$ we need only show the corresponding inclusion of domains:
\[
D\big((V_{i\pi}|_{\cP(I_+)\Omega})^*\big) \subseteq D(V_{i\pi}).
\]

Fix $\Phi' \in D\big((V_{i\pi}|_{\cP(I_+)\Omega})^*\big)$, and we must show that $\Phi' \in D(V_{i\pi})$. 
Let $\Phi \in \cP(I_+)\Omega$.
Since $\cP(I_+)\Omega \subset D(V_{i\pi})$, there exists a function $G:\bbS \to \bbC$ which is continuous, holomorphic in the interior $\mathring \bbS$, and satisfies for all $t \in \bbR$ 
\[
G(t) = \ip{V_t\Phi,\Phi'}, \qquad G(t+i\pi) = \ip{V_t V_{i\pi}\Phi, \Phi'},
\]
where as usual $\bbS=\{ \tau \in \bbC \, : \, 0 \le \Im(\tau) \le \pi\}$.
Let $\Phi'' = (V_{i\pi}|_{\cP(I_+)\Omega})^*\Phi'$.
Note that $V_t$ leaves $\cP(I_+)\Omega$ invariant and commutes with $V_{i\pi}$, so for $t \in \bbR$ we have
\[
G(t+i\pi)=\ip{V_{i\pi}V_{t}\Phi,\Phi'} = \ip{V_t\Phi, \Phi''} = \ip{\Phi, V_{-t}\Phi''}.
\]
Define $H_\Phi:\bbS \to \bbC$ by $H_\Phi(\tau) = \overline{G(-\overline{\tau})}$, and note that $H_\Phi$ is again continuous on $\bbS$ and holomorphic in the interior.
We have for $t \in \bbR$
\begin{equation}\label{eqn: H boundary}
H_\Phi(t) = \ip{V_t \Phi', \Phi}, \qquad H_\Phi(t+i\pi) = \ip{V_t \Phi'',\Phi}.
\end{equation}
In order to show that $\Phi' \in D(V_{i\pi})$, we must establish the existence of such a function $H_\Phi$ for all $\Phi \in \cD$, and we have done this so far when $\Phi \in \cP(I_+)\Omega$.

By Lemma~\ref{lem: X dense}, the closure of $\cP_\cX(I_+)\Omega$ contains $\cP(I_+)\Omega$, and thus $\cP_\cX(I_+)\Omega$ is dense in $\cD$ by the Reeh-Schlieder theorem \cite[App. A]{CarpiRaymondTanimotoTener25}.
Given arbitrary $\Phi \in \cD$, choose a net $\Phi_n \in \cP_\cX(I_+)\Omega$ such that $\Phi_n \to \Phi$.
We will show that the net $H_{\Phi_n}$ converges locally uniformly on $\bbS$ to the required function $H_\Phi$.
Our first step is to show that this net is locally uniformly Cauchy.

Note that $\overline{H_{\Phi_n}(-\overline{\tau})}$ is a function of the form $G$ considered above, and so by Lemma~\ref{lem: growth estimate on S} (and the fact that $\Phi_n \in \cP_\cX(I_+)\Omega$) we can infer that $e^{-2\tau^2}H_{\Phi_n}$ vanishes at infinity.
It follows that $e^{-2\tau^2}H_{\Phi_n}(\tau) - e^{-2\tau^2}H_{\Phi_m}(\tau)$ vanishes at infinity as well, and 
so the magnitude of
\begin{equation}\label{eqn: H Cauchy}
e^{-2\tau^2}H_{\Phi_n}(\tau) - e^{-2\tau^2}H_{\Phi_m}(\tau)
\end{equation}
is maximized on $\partial \bbS = \bbR \cup i\pi + \bbR$.
From \eqref{eqn: H boundary}, we see that the restrictions of \eqref{eqn: H Cauchy} to the components of $\partial \bbS$ are given by
\[
t \mapsto e^{-2t^2} \ip{V_t\Phi',\Phi_n-\Phi_m}, \qquad \mbox{and} \qquad t + i\pi \mapsto e^{2\pi^2 - 4\pi i t} e^{-2t^2}  \ip{V_t\Phi'',\Phi_n-\Phi_m},
\]
respectively.
Thus to show that \eqref{eqn: H Cauchy} converges uniformly to zero on $\bbS$, it suffices to show that for any $\Psi \in \cD$, functions of the form
\[
t \mapsto e^{-2t^2} \ip{V_t\Psi,\Phi_n-\Phi_m}
\]
converge uniformly to zero on $\bbR$ as $n,m \to \infty$.

It suffices to consider $\Psi=\varphi_1(f_1) \cdots \varphi_k(f_k)\Omega$ for $f_j \in C^\infty(S^1)$.
We have
\[
e^{-2t^2} V_t\Psi = \varphi_1(e^{-2t^2/k}\beta_{d_1}(v_t)f_1) \cdots \varphi_k(e^{-2t^2/k}\beta_{d_k}(v_t)f_k)\Omega.
\]
From the proof of Lemma~\ref{lem: growth estimate on R} (specifically \eqref{eqn: exp growth bound CN norm}), we have $e^{-2t^2/k}\beta_d(v_t)f \to 0$ in $C^\infty(S^1)$ as $t \to \pm \infty$.
Hence $e^{-2t^2}V_t\Psi \to 0$ in $\cD$, or in other words $e^{-2t^2}V_t \Psi$ extends to a continuous function $[-\infty,\infty] \to \cD$ which vanishes at infinity.
Since $[-\infty,\infty]$ is compact, the image of this map is a bounded subset of $\cD$, which means that for any open neighborhood $U$ of $0$ in $\cD$ there exists a scalar $M > 0$ such that $e^{-2t^2}V_t \Psi \in M U$ for all $t$.

Since the Hermitian form is jointly continuous, given $\epsilon > 0 $ we may choose a neighborhood $U_\epsilon$ of $0$ in $\cD$ such that $\abs{\ip{\Phi_1,\Phi_2}} < \epsilon$ when $\Phi_j \in U_\epsilon$.
Choose $M > 0$ as above so that $e^{-2t^2} V_t\Psi \in MU_\epsilon$ for all $t$, and note that $\abs{\ip{\Phi_1,\Phi_2}} < \epsilon$ when $\Phi_1 \in MU_\epsilon$ and $\Phi_2 \in M^{-1} U_\epsilon$.
Since $\Phi_n \to \Phi$, we may choose $n_0$ such that when $n,m \ge n_0$ we have $\Phi_n-\Phi_m \in M^{-1} U_\epsilon$.
We then have for all $t \in \bbR$
\[
\abs{e^{-2t^2} \ip{V_t\Psi,\Phi_n-\Phi_m} } < \epsilon,
\]
and so the net of functions $e^{-2t^2} \ip{V_t\Psi,\Phi_n-\Phi_m}$ converges uniformly to zero, as claimed.
As described above, it follows that $e^{-2\tau^2}(H_{\Phi_n} - H_{\Phi_m})$ converges uniformly to zero.
Hence $H_{\Phi_n}$ is a locally uniformly Cauchy net of functions, which therefore converges locally uniformly to a function $H_{\Phi}$ that is continuous on $\bbS$ and holomorphic in the interior.
Taking pointwise limits on the boundary we see that $H_{\Phi}$ satisfies \eqref{eqn: H boundary}.
Since $\Phi$ was arbitrary, we conclude that $\Phi' \in D(V_{i \pi})$, as claimed.
\end{proof}

As a consequence of Theorem~\ref{thm: Vipi selfadjoint}, we have $\big( V_{\pm i\pi}|_{\cP(I_\pm)\Omega}\big)^{**} = V_{\pm i\pi}$.
In the context of Hilbert spaces, the double adjoint coincides with the closure, and so the `moreover' of the Theorem may be thought of as a positive answer to a weakening of Question~\ref{ques: core}. 

As an aside, we remark that the above proof of Theorem~\ref{thm: Vipi selfadjoint} does not go through verbatim for the operator $\widetilde V_{i\pi}$, with the key obstruction being that linear functionals of the form $\ip{\, \cdot \, , \Phi}$ may not be dense in $\cD_\cF^*$ in the strong dual topology (i.e. the topology of uniform convergence on bounded sets), although they are dense in the weak-$*$ topology (i.e. the topology of pointwise convergence).

\section{Haag duality for Wightman CFTs with invariant Hermitian forms}\label{sec: nonunitary duality}

In this section we continue the notation of Section~\ref{sec: sesquilinear}, fixing an involutive Wightman CFT. 
We will now study the notion of `Haag duality' for nets of algebras, which is a key application of Tomita-Takesaki theory in the unitary setting. Our first step will be to make contact with the theory of standard subspaces, which have proven to be very powerful in analyzing duality phenomena in algebraic quantum field theory \cite{LongoLectureNotesI,LongoLectureNotesII} (see also Section~\ref{sec: background unitary} for the unitary case).
The main result of this section is the Haag duality of the double commutant net $\cP(I)''$.
The results of this section will also be applied in Section~\ref{sec: unitary} to gain insight into unitary models.

We first require a brief digression on topological notions.

\begin{defn}\label{defn: standard topologies}
    Let $\cD$ be a locally convex space with a nondegenerate Hermitian form.
    Then we refer to the given topology on $\cD$ as the \textbf{strong} topology, the weak topology on $\cD$ induced by the continuous dual $\cD^*$ as the \textbf{weak} topology, and the weak topology induced by the functionals $\{\ip{\, \cdot \, , \Phi} \, : \, \Phi \in \cD\}$ as the \textbf{$\cD$-weak} topology.
    We refer to these topologies as the \textbf{standard topologies} on $\cD$.
\end{defn}

In the context of an involutive Wightman theory, the strong topology is also known as the $\cF$-strong topology and the weak topology is also known as the $\cF$-weak topology.
The strong topology is stronger than the weak topology, which is stronger than the $\cD$-weak topology.
By the Hahn-Banach theorem, a convex subset of $\cD$ is strong closed if and only if it is weakly closed, but a strongly closed subspace of $\cD$ may not be $\cD$-weakly closed.
However, the results we will obtain below hold for any choice of standard topology.

\begin{defn}
    Let $\cD$ be a locally convex space with nondegenerate Hermitian form, equipped with one of its standard topologies.
    A closed real subspace $K \subset \cD$ is called \textbf{cyclic} (with respect to the topology) if $K + iK$ is dense in $\cD$, and called \textbf{separating} if $K \cap iK =\{0\}$.
    A closed real subspace is called \textbf{standard} if it is both cyclic and separating.
\end{defn}

To a standard subspace $K$, we can associate a densely defined conjugate-linear involution $S_K$ with domain $D(S_K) = K+iK$ given by $S_K(\Phi_1+i\Phi_2) = \Phi_1 - i \Phi_2$ for $\Phi_j \in K$.
The term involution here means that $S_K^2 = \mathrm{Id}_{D(S_K)}$.
The subspace $K$ is recovered as the fixed points of $S_K$.
The following is standard, although we include a proof to emphasize that unitarity is not required.

\begin{lem}\label{lem: std subspace vs involution}
    Let $\cD$ be a locally convex space with nondegenerate Hermitian form, equipped with one of its standard topologies.
Then the correspondence $K \leftrightarrow S_K$ induces a bijection between i) standard subspaces $K$ of $\cD$, and ii) densely-defined closed conjugate-linear involutions $S$.
\end{lem}
\begin{proof}
We need to check that $S_K$ is closed, and that every closed conjugate linear involution $S$ is of the form $S_K$.
To the first end, suppose that $\Phi_n \in D(S_K)=K+iK$, that $\Phi_n \to \Phi$ in $\cD$, and that $S_K\Phi_n \to \Psi$.
Let $\Phi_{1,n} = \frac{\Phi_n + S_K\Phi_n}{2}$ and $\Phi_{2,n} = \frac{\Phi_n - S_K \Phi_n}{2i}$, and note that $\Phi_n = \Phi_{1,n} + i \Phi_{2,n}$.
We then have $S_K \Phi_{j,n} = \Phi_{j,n}$, so that $\Phi_{j,n} \in K$.
Moreover $\Phi_{1,n} \to \frac{\Phi + \Psi}{2}$ and $\Phi_{2,n} \to \frac{\Phi - \Psi}{2i}$, and these vectors lie in $K$ since $K$ is closed.
Hence 
\[
\Phi = \lim \Phi_{1,n} + i \lim \Phi_{2,n} \in K + i K = D(S_K),
\]
and 
\[
S_K\Phi = \Psi = \lim \Phi_{1,n} - i \lim \Phi_{2,n} = \lim S_K \Phi_n.
\]
We conclude that $S_K$ is closed.

To complete the proof, we must show that given a closed densely-defined conjugate-linear involution $S$, the fixed points $K=\operatorname{Ker}(S-1)$ are a standard subspace such that $D(S)=K+iK$.
The argument is routine.
If $\Phi \in K \cap iK$, then applying $S$ yields $\Phi = - \Phi$ and so $\Phi = 0$ and $K$ is separating.
Moreover if $\Phi \in D(S)$ then $\frac{\Phi+S\Phi}{2}, \frac{\Phi-S\Phi}{2i} \in K$.
Hence $\Phi \in K + iK$, so $D(S) \subset K+iK$.
Since $D(S)$ is assumed dense, $K$ is cyclic.
Moreover $K$, and thus $K+iK$ are contained in $D(S)$ by construction, so we evidently have $D(S)=K+iK$ as required.
\end{proof}

\begin{defn}
    If $K \subset \cD$ is a real subspace, the \textbf{complementary subspace} $K'$ is defined by
    \[
    K' = \{ \Phi \in \cD \, : \ip{\Phi,\Psi} \in \bbR \, \text{ for all } \Psi \in K\}.
    \]
    This is the same as the perpendicular subspace to $iK$ with respect to the real form $\Re \ip{\, \cdot \, , \, \cdot \,}$.
\end{defn}

As described before Question~\ref{ques: core}, taking complements does not behave the same as in the Hilbert space case, as even when the form on $\cD$ is positive definite, the complement could be `unexpectedly small.'
For this reason, even when $K$ is standard the space $K'$ may fail to be cyclic, and even if $K'$ is standard we may have $K \subsetneq K''$.
We can, however, establish some basic results below in Lemma~\ref{lem: SK adjoint}.
First, the reader is reminded that the relation between a conjugate-linear operator $X$ and its adjoint $X^*$ is
\[
\ip{X\Phi,\Psi} = \ip{X^*\Psi,\Phi},
\]
and the definition of the domain $D(X^*)$ is adjusted accordingly.

\begin{lem}\label{lem: SK adjoint}
Let $\cD$ be a locally convex space with nondegenerate Hermitian form, equipped with one of its standard topologies.
Let $K \subset \cD$ be a closed real subspace.
\begin{enumerate}[(i)]
\item If $K$ is cyclic then $K'$ is separating, and if $K'$ is cyclic then $K$ is separating.
\item If $K$ is standard, then $K'$ is standard if and only if $S_K^*$ is densely defined, in which case $S_{K'} = S_K^*$.
\end{enumerate}
\end{lem}
\begin{proof}
    We first prove (i). Suppose $K$ is cyclic and $\Phi \in K' \cap iK'$.
    It follows that $\ip{\Phi,\Psi} \in \bbR \cap i \bbR$ for all $\Psi \in K$.
    Hence $\ip{\Phi,\Psi} = 0$ for all $\Psi \in K$, and indeed for all $\Psi \in K + i K$.
    Since $K$ is cyclic and the form is nondegenerate, we conclude that $\Phi=0$ and $K'$ is separating.
    On the other hand, if $K'$ is cyclic, then $K''$ is separating by the previous work, and so $K \subseteq K''$ is separating as well.

    We now prove (ii).
    Assume $K$ is standard, so by part (i) $K'$ is separating.
    Since $S_K$ is an involution, $D(S_K^*)$ is invariant under $S_K^*$ and $S_K^*$ is an involution as well.
    Let $\tilde K = \operatorname{ker}(S_K^*-1)$, and arguing as before yields $D(S_K^*) = \tilde K + i \tilde K$.
    Thus to prove (ii) it suffices to prove that $\tilde K = K'$.
    Indeed for $\Phi \in K$ and $\Psi \in \cD$ we have
    \[
    \ip{S_K\Phi,\Psi} = \ip{\Phi,\Psi} = \overline{\ip{\Psi,\Phi}}.
    \]
    Hence if $\Psi \in K'$ we see $\ip{S_K\Phi,\Psi} = \ip{\Psi,\Phi}$, which shows that $\Psi \in D(S_K^*)$ and $S_K^*\Psi = \Psi$, i.e., $\Psi \in \tilde K$.
    Conversely, if $\Psi \in \tilde K$ then $\ip{S_K\Phi,\Psi} = \ip{\Psi,\Phi}$, and we see that $\ip{\Psi,\Phi} \in \bbR$, so that $\Psi \in K'$.
    We conclude $K' = \tilde K$, completing the proof of (ii).
\end{proof}

We note that even when $K$ and $K'$ are both standard, we have not ruled out the possibility that $K \subsetneq K''$ and $S_K \subsetneq S_K^{**}$.
Typically one would expect this if $K$ were a standard subspace for the strong topology on $\cD$ but not $\cD$-weakly closed, as $K''$ would contain the $\cD$-weak closure.

We now return to our study of involutive Wightman CFTs $(\cF,\cD,U,\Omega)$.
\begin{defn}\label{def: standard subspaces from wightman}
    Given an interval $I$, we denote by $\cP(I)_{sa} \subset \cP(I)$ the real subspace of self-adjoint elements, and the corresponding real subspace of $\cD$ is denoted $\cK(I) := \overline{\cP(I)_{sa}\Omega}$, where the closure is taken in the strong (or equivalently weak) topology.
    We denote by $\cK_w(I)$ the closure of $\cK(I)$ in the $\cD$-weak topology.
\end{defn}

The reader is reminded that if $I \subset S^1$ is an interval, then $I'$ denotes the complementary interval.

\begin{lem}\label{lem: KI standard}
    For each interval $I$, the subspaces $\cK(I)$ and $\cK(I)'$ are standard subspaces of $\cD$, and $\cK(I') \subset \cK(I)'$.
    The same holds for the subspaces $\cK_w(I)$ and $\cK_w(I')$.
\end{lem}
\begin{proof}
    Observe that $\cK(I)$ contains $\cP(I)\Omega$, which is dense by the Reeh-Schlieder theorem, and so each $\cK(I)$ is cyclic.
    If $x \in \cP(I)_{sa}$ and $y \in \cP(I')_{sa}$, then locality for the net $\cP(I)$ implies that 
    \[
    \ip{x\Omega, y\Omega} = \ip{yx\Omega,\Omega} = \ip{xy\Omega,\Omega} = \ip{y\Omega,x\Omega} = \overline{\ip{x\Omega,y\Omega}}.
    \]
    Hence $\ip{x\Omega,y\Omega} \in \bbR$, and we see that $\cK(I') \subseteq \cK(I)'$.
    Since $\cK(I')$ is cyclic, so is $\cK(I)'$ whence $\cK(I)$ is separating by Lemma~\ref{lem: SK adjoint}.
    We have shown that $\cK(I)$ is standard and that $\cK(I)'$ is cyclic, and the fact that $\cK(I)'$ is separating again follows from Lemma~\ref{lem: SK adjoint}.

    In the case of $\cK_w$, note that $\cK(I)' = \cK_w(I)'$ and $\cK(I) \subset \cK_w(I)$.
    It immediately follows from the above that $\cK_w(I)'$ is standard, that $\cK_w(I') \subset \cK_w(I)'$, and that $\cK_w(I)$ is cyclic.
    The fact that $\cK_w(I)$ is separating follows again from Lemma~\ref{lem: SK adjoint}.
\end{proof}

From here it is straightforward to see that the spaces $\cK(I)$ and $\cK_w(I)$ each enjoy the following structure:

\begin{defn}\label{def: mob covariant net nonunitary}
    Let $\cD$ be a locally convex space equipped with a jointly continuous nondegenerate Hermitian form and a strongly continuous unitary representation $U:\Mob \to \End(\cD)$.
    A \textbf{(non-unitary) M\"obius-covariant net of standard subspaces} is a collection of standard subspaces $\cK(I) \subset \cD$ such that the following holds:
    \begin{enumerate}[1)]
        \item Isotony: If $I_1 \subseteq I_2$, then $\cK(I_1) \subseteq \cK(I_2)$.
        \item M\"obius covariance: $U(\gamma)\cK(I) = \cK(\gamma(I))$.
        \item Locality: $\cK(I') \subset \cK(I)'$.
    \end{enumerate}
\end{defn}

\begin{rem}
    Note that if the subspaces $\cK(I)$ are only assumed cyclic then standardness follows from locality, as in the proof of Lemma~\ref{lem: KI standard}.
    One could hope to replace cyclicity of each $\cK(I)$ with a weaker condition of being jointly cyclic. The question of whether this a priori weaker definition is equivalent is asking whether the Reeh-Schlieder phenomenon occurs, but we do not pursue this question.
    Without any positive energy assumption on $U$, one should not expect to deduce too much from axioms presented in Definition~\ref{def: mob covariant net nonunitary}; we include the definition to capture the structure we observe in our example.
\end{rem}

We have thus proven:

\begin{prop}
Let $(\cF,\cD,U,\Omega)$ be an involutive Wightman CFT on $S^1$.
Then the family of subspaces $\cK(I)=\overline{\cP(I)_{sa}\Omega}$ is a M\"obius-covariant net of standard subspaces of $\cD$ with respect to the representation $U$ of $\Mob$.
The same holds for the $\cD$-weak closures $\cK_w(I)$ of $\cK(I)$.
\end{prop}
\begin{proof}
    The standardness of the subspaces $\cK(I)$ and the locality axiom are Lemma~\ref{lem: KI standard}, and the other requirements are immediate from properties of the algebras $\cP(I)$.
\end{proof}

A key question which arises in this context, and has been extensively studied in various guises, is that of duality.

\begin{defn}
    Let $\cD$ be a locally convex space with nondegenerate Hermitian form.
    A M\"obius-covariant net of standard subspaces satisfies \textbf{Haag duality} if $\cK(I')=\cK(I)'$ for all intervals $I$.
    We call the net \textbf{weakly Haag dual} if the $\cD$-weak closures $\cK_w(I) = \overline{\cK(I)}$  are Haag dual.
\end{defn}

In the context of positive energy nets of standard subspaces of a Hilbert space, Haag duality is automatic \cite[Thm. 3.3.1]{LongoLectureNotesI}.
We now investigate the question of Haag duality for the nets of standard subspaces $\cK(I)$ and $\cK_w(I)$ arising from a Wightman CFT (as in Definition~\ref{def: standard subspaces from wightman}).

\begin{lem}\label{lem: dual standard subspaces}
\text{}

    \begin{enumerate}[(i)]
    \item As unbounded operators, $\theta V_{\pm i\pi} \theta = V_{\mp i \pi}$ and $V_{i\pi}^{-1} = V_{-i\pi}$.
    \item The operators $S_\pm = \theta V_{\pm i \pi}$ are closed conjugate-linear involutions (with respect to every standard topology) and $S_{\pm}^* = S_{\mp}$.
    \item If $\cL_{\pm}$ is the standard subspace corresponding to $S_{\pm}$, then $\cL_{\pm}' = \cL_{\mp}$, $\cK(I_\pm) \subset \cK_w(I_\pm) \subset \cL_{\pm}$, and $V_t\cL_{\pm} = \cL_{\pm}$ for all $t \in \bbR$.
    \item $\cK(I_{\pm})' = \cK_w(I_{\pm})' = \cL_\mp$
    \end{enumerate}
\end{lem}
\begin{proof}
    We first show that $\theta$ maps $D(V_{i \pi})$ into $D(V_{-i\pi})$.
    Indeed if $\Phi \in D(V_{i\pi})$ and $\Phi' \in \cD$, then we have a map $G_{\Phi'}:\bbS_{i\pi} \to \bbC$ that is continuous and holomorphic in the interior, and
    \begin{equation}\label{eqn: GPhiprime}
    G_{\Phi'}(t) = \ip{V_t\Phi, \Phi'}, \qquad G_{\Phi'}(t+i\pi) = \ip{V_tV_{i\pi} \Phi,\Phi'}.
    \end{equation}
    Define $H_{\Phi'}:\bbS_{-i\pi} \to \bbC$ by $H_{\Phi'}(\tau) = \overline{G_{\theta \Phi'}(\overline{\tau})}$.
    Since $\theta$ commutes with $V_t$, we have
    \[
    H_{\Phi'}(t) = \overline{\ip{V_t\Phi,\theta\Phi'}} = \ip{V_t\theta\Phi,\Phi'}
    \]
    and similarly
    \[
    H_{\Phi'}(t-i\pi) = \overline{\ip{V_tV_{i\pi}\Phi,\theta \Phi'}} = \ip{V_t\theta V_{i\pi} \Phi,\Phi'}.
    \]
    Hence $\theta\Phi \in D(V_{-i\pi})$ and $V_{-i\pi}\theta\Phi = \theta V_{i\pi}\Phi$.
    The same argument shows that if $\Phi \in D(V_{-i\pi})$ then $\theta \Phi \in D(V_{i\pi})$ and $\theta V_{-i\pi}\Phi = V_{i\pi}\theta \Phi$.
    Since $\theta$ is an involution and maps $D(V_{\pm i \pi})$ into $D(V_{\mp i \pi})$, we actually must have $\theta D(V_{\pm i \pi}) = D(V_{\mp i \pi})$, and so from our previous calculations $\theta V_{\pm i \pi} \theta = V_{\mp i \pi}$.

    We next show that $V_{i\pi}^{-1} = V_{-i\pi}$, or equivalently that $V_{\mp i\pi} V_{\pm i\pi} = \mathrm{Id}_{D(V_{\pm i\pi})}$. We will show $V_{-i\pi}V_{i\pi} = \mathrm{Id}_{D(V_{i\pi})}$, and the other identity is similar.
    Let $\Phi \in D(V_{i\pi})$ and $\Phi' \in \cD$, and let $G_{\Phi'}$ be as in \eqref{eqn: GPhiprime}.
    We now define $\widetilde H_{\Phi'}:\bbS_{-i\pi} \to \bbC$ by $\widetilde H_{\Phi'}(\tau) = G_{\Phi'}(\tau + i\pi)$.
    Then for $t \in \bbR$ we have
    \[
     \widetilde H_{\Phi'}(t) = \ip{V_tV_{i\pi} \Phi,\Phi'}, \qquad \widetilde H_{\Phi'}(t-i\pi) = \ip{V_t\Phi, \Phi'}.
    \]
    Since $\Phi'$ was arbitrary, this demonstrates that $V_{i\pi}\Phi \in D(V_{-i\pi})$ and that $V_{-i\pi}V_{i\pi}\Phi = \Phi$, as claimed. This completes the proof of (i).
    
    From part (i), we see that $S_{\pm}^2 = V_{\mp i \pi} V_{\pm i \pi} = \mathrm{Id}_{D(V_{\pm i \pi})} = \mathrm{Id}_{D(S_{\pm})}$, and so $S_{\pm}$ are conjugate-linear involutions.
    The maps $V_{\pm i\pi}$ are closed for the $\cD$-weak topology (and so for every standard topology) since adjoints are automatically $\cD$-weakly closed and $V_{\pm i\pi}$ are self-adjoint by Theorem~\ref{thm: Vipi selfadjoint}. 
    Hence $S_{\pm}$ are closed as well, and we have
    \[
    S_{\pm}^* = V_{\pm i \pi}\theta = \theta V_{\mp i \pi} = S_{\mp}
    \]
    where we note that the first equality is not immediate but easy to verify. This completes (ii).

    The identity $\cL_{\pm}' = \cL_{\mp}$ now follows from Lemma~\ref{lem: SK adjoint}.
    By definition $\cL_{\pm}$ consists of fixed points of $S_{\pm}$, and so $\cP(I_{\pm})_{sa}\Omega \subset \cL_{\pm}$ by Theorem~\ref{thm: involutive BW}.
    Taking closures yields $\cK(I_{\pm}) \subset \cL_{\pm}$.
    Finally note that $V_t$ commutes with both $\theta$ and $V_{\pm i \pi}$, and so with $S_{\pm}$.
    Hence $V_t$ leaves $\cL_{\pm}$ invariant for all $t$, which implies that $V_t\cL_{\pm} = \cL_{\pm}$, completing the proof of (iii).

    To prove (iv), it suffices to show that $\cK(I_\pm)' = \cL_\mp$, and by (iii) we only need to check $\cK(I_\pm)' \subseteq \cL_\mp$. Without loss of generality we only consider the inclusion $\cK(I_+)' \subseteq \cL_-$.
    Let $\Phi \in \cK(I_+)'$. Then for all $\Psi \in \cK(I_+)$ we have
    \[
    \ip{\Psi,\Phi} = \ip{\Phi,\Psi} = \ip{\Phi,S_+\Psi} = \ip{\Phi,\theta V_{i\pi}\Psi} = \ip{V_{i\pi}\Psi,\theta \Phi}
    \]
    where the first equality is the definition of $\cK(I_+)'$, the second is the fact that $\cK(I_+) \subset \cL_+$, the third is the definition of $S_+$, and the fourth is the fact that $\theta$ is an antiunitary involution.
    The identity 
    \begin{equation}\label{eqn: proof KIprime}
    \ip{\Psi,\Phi} = \ip{V_{i\pi}\Psi,\theta\Phi}
    \end{equation}
    extends by linearity to all $\Psi \in \cK(I_+) + i \cK(I_+)$, and in particular \eqref{eqn: proof KIprime} holds for all $\Psi \in \cP(I_+)\Omega$.
    Hence $\theta\Phi$ lies in the domain of $(V_{i\pi}|_{\cP(I_+)\Omega})^*$, and we recall from Theorem~\ref{thm: Vipi selfadjoint} that $(V_{i\pi}|_{\cP(I_+)\Omega})^* = V_{i\pi}$.
    Thus from \eqref{eqn: proof KIprime} we may deduce that $\theta \Phi \in D(V_{i\pi})$ and that $\ip{\Psi,\Phi} = \ip{\Psi,V_{i\pi}\theta\Phi}$ for all $\Psi \in \cP(I_+)\Omega$.
    We then have
    \[
    \Phi = V_{i\pi}\theta \Phi = \theta V_{-i\pi}\Phi = S_- \Phi.
    \]
    where the first equality uses the Reeh-Schlieder property and the nondegeneracy of the form, the second uses part (i) and implicitly asserts that $\Phi \in D(V_{-i\pi})$, and the third is the definition of $S_-$.
    Hence by definition we have $\Phi \in \cL_-$, as required.
\end{proof}

As a consequence of Lemma~\ref{lem: dual standard subspaces}, we can show that the `dual net' $I \mapsto \cK(I')'$ is local, and moreover satisfies Haag duality.

\begin{cor}\label{cor: haag dual subspace dual net}
The net of subspaces $\cL(I):= \cK(I')'$ is a M\"obius-covariant net of standard subspaces of $\cD$, and moreover this net satisfies Haag duality.
\end{cor}
\begin{proof}
    Isotony and M\"obius-covariance of the net $\cL(I)$ follow from the corresponding properties of the net $\cK(I)$.
    By M\"obius covariance, it suffices to check locality and Haag duality for the pair of intervals $I_\pm$.
    Observe that $\cL(I_\pm)$ is the space $\cL_\pm$ of Lemma~\ref{lem: dual standard subspaces}, by part (iv) of that lemma.
    We have $\cL_\pm' = \cL_\mp$ by part (iii) of the same lemma, which immediately implies locality and Haag duality.
\end{proof}

\begin{rem}\label{rem: inclusion of standard subspaces}
We have $\cK(I) \subseteq \cL(I)$ by the locality of $\cK$, and evidently the Haag duality of the net $\cK$ is equivalent to this inclusion being an equality.
By Haag duality of the net $\cL$, we have $\cL(I)=\cK(I)''$, so the question of Haag duality of the net $\cK$ rests on the difference between taking the closure of a subspace versus taking the double complement.
The question of Haag duality for $\cK$ can also be interpreted  in the context of inclusions of standard subspaces $\cK(I) \subset \cL(I)$.
The involution $S_{\pm}$ corresponding to $\cL(I_\pm)$ satisfies $S_{\pm} = \theta V_{\pm i\pi}$, and $\cK(I_{\pm})$ is invariant under $V_t$.
In the unitary setting the identity $\cK(I_{\pm}) = \cL(I_\pm)$ would follow (see e.g. \cite[Prop. 2.1.10]{LongoLectureNotesI}), but we do not see an immediate generalization to our context.
\end{rem}

In light of Lemma~\ref{lem: dual standard subspaces}, it is also possible to rephrase the question of Haag duality for $\cK$ in operator theoretic language.

\begin{cor}\label{cor: haag dual subspace net}
The net of standard subspaces $\cK(I) = \overline{\cP(I)_{sa}\Omega}$ satisfies Haag duality if and only if $\cP(I_+)\Omega$ is a core for $V_{i\pi}$ for the $\cF$-strong topology.
Similarly $\cK_w(I)$ satisfies Haag duality if and only if $\cP(I_+)\Omega$ is a core for $V_{i\pi}$ for the $\cD$-weak topology.
\end{cor}
\begin{proof}
If $\cP(I_+)\Omega$ is a $\cF$-strong core for $V_{i\pi}$ then $\cK(I_+) = \cL(I_+)$.
By M\"obius-covariance we have $\cK(I) = \cL(I)$ for all intervals $I$, and Haag duality follows.
For the converse, suppose the net $\cK$ is Haag dual, in which case $\cK(I_+) = \cK(I_-)' = \cL(I_+)$.
It follows that $\cP(I_+)\Omega$ is a $\cF$-strong core for $V_{i\pi}$.
The proof for the $\cD$-weak topology is similar.
\end{proof}

We conclude this section with a study of Haag duality for nets of algebras.
As with standard subspaces, we have a preliminary definition of a local net of algebras (for which one does not expect a robust general theory, in part because of the lack of a positive energy condition).

\begin{defn}\label{def: mob covariant net algebras nonunitary}
    Let $\cD$ be a locally convex space equipped with a jointly continuous nondegenerate Hermitian form, a strongly continuous unitary representation $U:\Mob \to \End(\cD)$, and a choice of non-zero vector $\Omega \in \cD$.
    A \textbf{(non-unitary) M\"obius-covariant net of algebras} is a collection of unital $*$-subalgebras $\cQ(I) \subset \End_*(\cD)$ such that the following holds:
    \begin{enumerate}[1)]
        \item Isotony: If $I_1 \subseteq I_2$, then $\cQ(I_1) \subseteq \cQ(I_2)$.
        \item M\"obius covariance: $U(\gamma)\cQ(I)U(\gamma)^{-1} = \cQ(\gamma(I))$.
        \item Locality: $\cQ(I)$ and $\cQ(I')$ commute for all intervals $I$.
        \item Vacuum: the vector $\Omega$ is fixed by $\Mob$ and is cyclic for each algebra $\cQ(I)$. 
    \end{enumerate}
    The net is called \textbf{Haag dual} if $\cQ(I)' = \cQ(I')$.
    Here we have used the notation that if $S \subset \End_*(\cD)$, then $S'$ denotes the commutant, consisting of all operators in $\End_*(\cD)$ which commute with $S$ elementwise.
\end{defn}

It is straightforward to check that if $\cQ(I)$ is a net of algebras as above, then $\cK(I) = \overline{\cQ(I)_{sa}\Omega}$ is a net of standard subspaces as in Definition~\ref{def: mob covariant net nonunitary}.

Returning to our Wightman CFT, the net of algebras $\cP(I)$ evidently satisfies the conditions of Definition~\ref{def: mob covariant net nonunitary}.
Indeed, the net $\cQ(I):=\cP(I)''$ is also a M\"obius-covariant  net of algebras (the only condition to check is locality, which follows formally from taking  two commutants of the inclusion $\cP(I) \subset \cP(I')'$).

We will shortly verify the Haag duality property for the net of algebras $\cQ(I)$. First, however, we isolate an important technical step.
The following lemma is analogous to standard computations in the context of Tomita-Takesaki theory and Haag duality for unitary theories. 

\begin{lem}\label{lem: algebra commutant inclusion}
Let $\cD$ be a locally convex vector space equipped with a jointly continuous nondegenerate Hermitian form and let $\Omega \in \cD$.
Let $A,B \subset \End_*(\cD)$ be unital $*$-subalgebras, and let $\theta$ be an antiunitary involution on $\cD$ such that $\theta \Omega = \Omega$.
Suppose that the following hold:
\begin{enumerate}[(i)]
\item Suppose that for all $X,Y \in A_{sa}$ we have $\ip{\theta X\Omega,Y\Omega} \in \bbR$.
\item Suppose $B\Omega$ is $\cD$-weakly dense in $\cD$,  that $B$ commutes with $A$, and that $\theta B \theta \subset A$.
\end{enumerate}
Then we have $\theta A\theta \subset A'$.
\end{lem}
\begin{proof}
    Let $X,Y \in A_{sa}$, so that by assumption we have $\ip{\theta X\Omega, Y\Omega} = \ip{Y\Omega, \theta X\Omega}.$
Since $\theta\Omega = \Omega$, this becomes
\[
\ip{\theta X \theta Y \Omega, \Omega} = \ip{Y\theta X \theta \Omega, \Omega}.
\]
Both sides of the above identity are conjugate linear in $X$ and linear in $Y$, so it extends to all $X,Y \in A$.
If $S,T \in B$, then $\theta S \theta, \theta T \theta \in A$.
Thus we apply the above identity with $X$ replaced by $\theta T \theta X\theta S \theta$ to obtain
\[
\ip{T \theta X \theta S Y \Omega, \Omega} = \ip{YT\theta X \theta S \Omega, \Omega}
\]
for all $X,Y \in A$ and $S,T \in B$.
Since $S$ and $T$ commute with $Y$, this becomes
\[
\ip{(\theta X \theta)Y S\Omega, T^*\Omega} = \ip{Y(\theta X \theta)S\Omega, T^*\Omega}.
\]
Now for arbitrary $\Phi' \in \cD$ choose a net $T_m^*\Omega \in B\Omega$ such that $T_m^*\Omega \to \Phi'$ in the $\cD$-weak topology.
Taking the limit of the identity
\[
\ip{(\theta X \theta)Y S\Omega, T^*_m\Omega} = \ip{Y(\theta X \theta)S\Omega, T^*_m\Omega}.
\]
yields 
\[
\ip{(\theta X \theta)Y S\Omega, \Phi'} = \ip{Y(\theta X \theta)S\Omega, \Phi'}.
\]
Now for arbitrary $\Phi \in \cD$ choose a net $S_n\Omega \in B\Omega$ such that $S_n\Omega \to \Phi$.
We now have
\[
\ip{ S_n\Omega, [(\theta X \theta)Y]^*\Phi'} = \ip{(\theta X \theta)Y S_n\Omega, \Phi'} = \ip{Y(\theta X \theta)S_n\Omega, \Phi'} = \ip{S_n\Omega, [Y(\theta X \theta)]^*\Phi'}.
\]
Taking the limit yields
\[
\ip{(\theta X \theta)Y \Phi, \Phi'} = \ip{Y(\theta X \theta)\Phi, \Phi'}.
\]
Since $\Phi,\Phi' \in \cD$ were arbitrary, we see that $\theta X \theta$ and $Y$ commute.
Hence $\theta X \theta \in A'$, which completes the proof.
\end{proof}

We can now prove Haag duality of the net $\cQ(I)$:

\begin{thm}[Haag duality]\label{thm: algebra haag duality}
    Let $(\cF,\cD,U,\Omega)$ be an involutive Wightman CFT on $S^1$, and let $\cP(I) \subset \End_*(\cD)$ be the unital algebra generated by operators $\varphi(f)$, where $\varphi \in \cF$ and $\supp f \subset I$.
    Let $\cQ(I)$ be the M\"obius-covariant net of algebras $\cQ(I):=\cP(I)''$.
    Then $\cQ$ satisfies the Haag duality condition: $\cQ(I)' = \cQ(I')$ for all intervals $I$, where the commutant is taken in $\End_*(\cD)$.
    Moreover, $\cQ(I) = \cP(I')'$.
\end{thm}
\begin{proof}
By M\"obius covariance, it suffices to consider the interval $I=I_+$.
Observe that is suffices to prove the ``moreover'' statement $\cQ(I_+)=\cP(I_-)'$, as taking commutants in this identity immediately yields the Haag duality of $\cQ$.
By locality of the net $\cP$, we have $\cP(I_-) \subseteq \cP(I_+)'$, and taking commutants yields $\cQ(I_+) = \cP(I_+)'' \subseteq \cP(I_-)'$.
We therefore need only prove the reverse inclusion $\cP(I_-)' \subseteq \cP(I_+)''$.
We know that $\cP(I_-) = \theta \cP(I_+) \theta$, and taking commutants yields $\cP(I_-)' = \theta \cP(I_+)' \theta$.
Thus the theorem reduces to establishing that $\theta \cP(I_+)' \theta \subset \cP(I_+)''$.

We will establish this last inclusion by invoking Lemma~\ref{lem: algebra commutant inclusion} with $A=\cP(I_+)'$ and $B = \cP(I_+)$.
The conclusion of the lemma is exactly the required inclusion, so we just verify that the two hypotheses hold.
We first consider hypothesis (i) of the Lemma.
Note that if $X \in \cP(I_+)'$ and $X=X^*$, then for all $Y \in \cP(I_+)_{sa}$ we have $\ip{X\Omega,Y\Omega} \in \bbR$, and so $X\Omega \in \cK(I_+)'$.

If we set $\cL(I_\pm)=\cK(I_\mp)'$, then we have just shown that $(\cP(I_+)')_{sa}\Omega \subset \cL(I_-)$.
A similar calculation shows that $(\cP(I_-)')_{sa}\Omega \subset \cL(I_+)$.
By Corollary~\ref{cor: haag dual subspace dual net}, the net $\cL(I)$ is local, and so in particular the pair of subspaces $(\cP(I_+)')_{sa}\Omega$ and $(\cP(I_-)')_{sa}\Omega$ are local as well.
If $X,Y \in (\cP(I_+)')_{sa}$, then $\theta X \theta \in \cP(I_-)'$ and it now follows that
\[
\ip{\theta X\Omega, Y\Omega} \in \bbR
\]
which establishes condition (i) of the Lemma.
Condition (ii) is straightforward, so the Lemma says that $\theta \cP(I_+)'\theta \subset \cP(I_+)''$, which completes the proof.
\end{proof}

\section{Applications for unitary CFTs}\label{sec: unitary}

A \textbf{unitary} Wightman CFT $(\cF,\cD,U,\Omega)$ is an involutive Wightman CFT with antiunitary PCT operator $\theta$ (see the opening of Section~\ref{sec: sesquilinear}) for which the invariant Hermitian form is positive definite, i.e. an inner product, and normalized so that $\norm{\Omega}=1$.
In this case $\cD$ embeds continuously in its Hilbert space completion $\cH$ (see \cite[\S2.3]{RaymondTanimotoTener22}).

In addition to the standard topologies ($\cF$-strong, $\cF$-weak, and $\cD$-weak; see Definition~\ref{defn: standard topologies}), the space $\cD$ has a norm topology.
The relative strength of these topologies is summarized in the diagram below, with arrows going from stronger topologies to weaker topologies: 
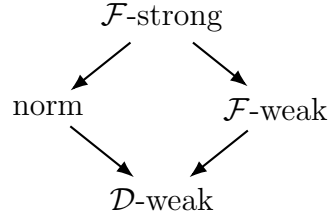
\[
\begin{tikzpicture}[arrow/.style={-{Latex}, thick}]
  \node (E) at (0,0) {$\cF$-strong};
  \node (C) at (-1.5,-1.2) {norm};
  \node (D) at (1.5,-1.2) {$\cF$-weak};
  \node (B) at (0,-2.4) {$\cD$-weak};

  \draw[arrow] (E) -- (C);
  \draw[arrow] (E) -- (D);
  \draw[arrow] (C) -- (B);
  \draw[arrow] (D) -- (B);
\end{tikzpicture}
\]
As usual, let $\cP(I) \subset \End_*(\cD)$ be the unital $*$-algebra generated by $\varphi(f)$ with $\operatorname{supp} f \subset I$, and let $\cP(I)_{sa}$ be the real subspace of self-adjoint elements.
In Section~\ref{sec: nonunitary duality} we defined $\cK(I)$ and $\cK_w(I)$ to be the closure of $\cP(I)_{sa}\Omega$ in the $\cF$-strong and $\cD$-weak topologies, respectively.
Let $\cK_n(I)$ be the closure in $\cD$ of $\cP(I)_{sa}\Omega$ in the norm topology, and let $\hat \cK_n(I)$ be the closure of $\cP(I)_{sa}\Omega$ in $\cH$.
We have
\[
\cK(I) \subset \cK_n(I) \subset \cK_w(I), \qquad \cK_n(I) \subset \hat \cK_n(I).
\]

\begin{prop}\label{prop: unitary K weakly haag dual}
    The net $\hat \cK_n(I)$ is a Haag dual M\"obius covariant net of standard subspaces of $\cH$.
    The net $\cK_n(I)$ is a Haag dual M\"obius covariant net of standard subspaces of $\cD$, and $\cK_n(I)=\cK_w(I)=\cK(I)''$.
\end{prop}
\begin{proof}
    The locality, M\"obius covariance, and isotony of the net $\hat \cK_n$ is inherited from $\cK$.
    Haag duality for this net then follows by general theory \cite[Thm. 3.3.1]{LongoLectureNotesI}.
    Haag duality for $\cK_n$ then follows from the fact that $\cK_n(I) = \hat \cK_n(I) \cap \cD$ and that $\cK_n(I)$ is dense in $\hat \cK_n(I)$ for each interval $I$.
    Since $\cK_n(I)$ is $\cD$-weakly dense in $\cK_w(I)$, we have $\cK_w(I)' = \cK_n(I)' = \cK_n(I') \subseteq \cK_w(I')$. Locality for $\cK_w$ furnishes the opposite inclusion $\cK_w(I') \subset \cK_w(I)'$, and we conclude that $\cK_w$ is Haag dual and $\cK_n = \cK_w$.
    Since $\cK(I)$ is $\cD$-weakly dense in $\cK_w(I)$, we have $\cK(I)' = \cK_w(I)' = \cK_w(I')$, and by Haag duality for $\cK_w$ it follows that $\cK(I)'' = \cK_w(I)$.
\end{proof}

Recall from Section~\ref{sec: background unitary} that given a standard subspace $K$ of a Hilbert space $\cH$, we have a corresponding conjugate linear closed involution $S_K$ with polar decomposition $S_K = J\Delta^{1/2}$.
Given a M\"obius covariant net of standard subspaces $\cK(I) \subset \cH$, covariant with respect to the unitary representation $U$ of $\Mob$ on $\cH$, then the Bisognano-Wichmann property for $\cK$ \cite[Thm. 3.3.1]{LongoLectureNotesI} states that the modular automorphisms $\Delta_\pm^{it}$ for $\cK(I_\pm)$ are given by
\begin{equation}\label{eqn: Delta pm it}
\Delta_{\pm}^{it} = U(v_{\mp 2\pi t}),
\end{equation}
where $v_t$ is the one-parameter subgroup of $\Mob$ fixing $I_{\pm}$ defined in Section~\ref{sec: background unitary}.
In Definition~\ref{defn: Vs involutive}, we considered densely defined operators $V_{\pm i \pi}$ on $\cD$ which were analytic continuations of the group $U(v_{t})|_{\cD}$.
It is therefore not surprising that the modular operator $\Delta_{\pm}^{1/2}$ is connected to $V_{\pm i\pi}$, and in fact we have an identity of unbounded operators.

\begin{thm}\label{thm: modular for unitary}
    Let $(\cF,\cD,U,\Omega)$ be a unitary Wightman CFT with PCT operator $\theta$, and let $\cH$ be the Hilbert space completion of $\cD$.
    Let $\hat \cK_n(I) \subset \cH$ be the corresponding net of standard subspaces defined as the norm closure of $\cP(I)_{sa}\Omega$ in $\cH$.
    Let $V_{\pm i\pi}$ be the densely defined operators on $\cD$ defined in Definition~\ref{defn: Vs involutive}.
    Let 
    \[
    \cE_\pm = \{ \Phi \in D(\Delta^{1/2}_{\pm}) \, : \, \Phi \in \cD \text{ and } \Delta_{\pm}^{1/2}\Phi \in \cD\}.
    \]
    Then the modular operator and conjugation for $\hat \cK_n(I_{\pm})$ satisfy
    \[
    \Delta_{\pm}^{1/2}|_{\cE_{\pm}} = V_{\pm i\pi}, \qquad J_{\pm} = \theta.
    \]
\end{thm}
\begin{proof}
    We first show that $\cE_{\pm}$ is contained in the domain of $V_{i\pi}$, and that $V_{i\pi}$ and $\Delta_+^{1/2}$ agree on this space.
    Without loss of generality we only consider $\cE_+$.
    Given $\Phi \in D(\Delta_+^{1/2})$, the map $F:\bbS \to \cH$ given by
    \[
    F(\tau) = \Delta_+^{\tfrac{\tau}{2\pi i}} \Phi
    \]
    is continuous on the closed strip and holomorphic on the interior.
    Hence for each $\Phi' \in \cD$, the map $F_{\Phi'}:\bbS \to \bbC$ given by $F_{\Phi'}(\tau) = \ip{F(\tau),\Phi'}$ is again continuous on the closed strip and holomorphic in the interior.
    Since $\Delta^{it}_+ = V_{-2\pi t}$ by \eqref{eqn: Delta pm it}, it follows that
    \[
    F_{\Phi'}(t) = \ip{V_t \Phi, \Phi'}, \qquad F_{\Phi'}(t + i \pi) = \ip{V_t \Delta_+^{1/2}\Phi,\Phi'}.
    \]
    Hence if $\Phi, \Delta^{1/2}_+\Phi \in \cD$, then it follows that $\Phi \in D(V_{i\pi})$ and $V_{i\pi}\Phi = \Delta_+^{1/2}\Phi$.

    It remains to check that $\cE_{\pm} = D(V_{\pm i\pi})$, and since we have shown $\cE_{\pm} \subset D(V_{\pm i \pi})$, it only remains to check the opposite inclusion.
    As in Lemma~\ref{lem: dual standard subspaces}, let $\cL_{\pm} \subset \cD$ be the standard subspace corresponding to the involution $S_{\pm} = \theta V_{\pm i \pi}$, so that $D(V_{\pm i \pi}) = \cL_{\pm} + i \cL_{\pm}$.
    By the same lemma, $\cL_{\pm}' = \cL_{\mp}$ and $\cK_w(I_{\pm}) \subset \cL_{\pm}$.
    By Proposition~\ref{prop: unitary K weakly haag dual}, the net $\cK_w(I_{\pm})$ is Haag dual, so taking complements yields $\cL_{\mp} \subset \cK_w(I_{\mp})$.
    Thus we have shown $\cL_{\pm} = \cK_w(I_{\pm})$.
    By Proposition~\ref{prop: unitary K weakly haag dual} we have $\cK_w(I_{\pm}) = \cK_n(I_{\pm}) = \hat \cK_n(I_{\pm}) \cap \cD$.
    Putting it all together, we conclude that $\cL_{\pm} = \hat \cK_n(I_{\pm}) \cap \cD$.
    We thus have:
    \[
    D(V_{\pm i\pi}) = \cL_{\pm} + i \cL_{\pm} \subset \hat \cK_n(I_{\pm}) + i \hat \cK_n(I_{\pm}) = D(\Delta^{1/2}_{\pm}).
    \]
    Arguing as in the first half of the proof, we see that if $\Phi \in D(V_{\pm i\pi})$ then $\Delta^{1/2}_{\pm}\Phi = V_{\pm i\pi}\Phi \in \cD$.
    Hence we have that $D(V_{\pm i\pi}) \subset \cE_{\pm}$, completing our proof that $D(V_{\pm i\pi}) = \cE_{\pm}$.

    We finally compute the modular conjugation $J_+=J_-$. Let $\hat S_+$ be the involution corresponding to $\hat \cK_n(I_+)$, so that $\hat S_+ = J_+ \Delta_+^{1/2}$.
    If $x \in \cP(I_+)_{sa}$, then since $x\Omega \in \cK(I_+) \subset \hat \cK_n(I_+)$ we have
    \[
    J_+ \Delta_+^{1/2} x\Omega = \hat S_+ x\Omega = x\Omega = \theta V_{i\pi} x \Omega = \theta \Delta_+^{1/2} x\Omega.
    \]
    The first equality is the definition of $J_+$ and $\Delta_+^{1/2}$, the second is the fact that $x\Omega \in \hat \cK_n(I_+)$, the third is Theorem~\ref{thm: involutive BW}, and the fourth follows from the previous part of this proof.
    Hence $J_+$ and $\theta$ agree on $\Delta^{1/2}_+ x\Omega$ when $x=x^*$, and by antilinearity this holds for all $x \in \cP(I_+)$.
    Hence $\theta$ and $J_+$ agree on $\cP(I_-)\Omega$, which is dense by Reeh-Schlieder.
    Since both are antiunitary, we conclude $\theta = J_+$.
\end{proof}

We now turn our attention to the correspondence between vertex algebras (i.e. Wightman CFTs) and conformal nets, which were introduced in Section~\ref{sec: background unitary}.
Given a unitary M\"obius vertex algebra $\cV$, we embed $\cV$ in its Hilbert space completion $\cH$, and for $v \in \cV$ with conformal dimension $d$ and $f \in C^\infty(S^1)$ we define the smeared field operator $Y(v,f)$ to be the closure of the operator $Y^0(v,f)$ defined on $\cV$ by
\[
Y^0(v,f)u = \sum_{n \in \bbZ} \hat f(n) v_{(n+d-1)}u, \qquad u \in \cV.
\]
For analytic details, see \cite[\S2.4]{RaymondTanimotoTener22} (and also \cite{CKLW18} for a treatment in the presence of energy bounds).
We now let $\cF$ be the Wightman CFT associated with the generating set of \emph{all} quasiprimary vectors $v$, i.e. 
\[
\cF = \{ Y(v, \cdot) \, : \, v \text{ quasiprimary} \},
\]
whose domain $\cD$ embeds continuously in the Hilbert space completion $\cH$ of $\cV$.
Since $\cV$ is unitary, it is semisimple as a representation of $\mathfrak{sl}_2(\bbC)$ and therefore generated by quasiprimary fields.
Following \cite{RaymondTanimotoTener22} (and, more generally, \cite{CKLW18} when the vertex algebra is energy bounded) we have the following notion of integrability for $\cV$.
Recall that if $S$ is a set of closed operators on a Hilbert space $\cH$, the von Neumann algebra  generated by $S$ is the von Neumann algebra generated by the polar partial isometries and spectral projections of all operators $X \in S$.

\begin{defn}\label{defn: AQFT-local]}
    Let $\cV$ be a unitary M\"obius vertex algebra with Hilbert space completion $\cH$.
    The local algebra $\cA_\cV(I)$ associated with an interval $I \subset S^1$ is the von Neumann algebra generated by all operators $Y(v,f)$ with $v$ quasiprimary and $f \in C^\infty(S^1)$ with $\operatorname{supp} f \subset I$.
    We say that $\cV$ is \textbf{AQFT-local} if $\cA_\cV(I)$ and $\cA_\cV(J)$ commute element-wise when $I$ and $J$ are disjoint.
\end{defn}

We use the term AQFT-local rather than the term `strongly local' from \cite{CKLW18} as the latter definition also assumes `energy bounds,' which are not required here.
It is reasonably straightforward to show that if $\cV$ is AQFT-local, then $\cA_\cV$ is a M\"obius-covariant Haag-Kastler net \cite[Prop. 4.8]{RaymondTanimotoTener22} (which is a reproof of a result of \cite[Thm. 6.8]{CKLW18}, without energy bounds).

It is believed that every unitary M\"obius vertex algebra is AQFT-local \cite[Conj. 8.18]{CKLW18}.
There is no criterion for determining when the von Neumann algebras generated by two families of closed operators commute which is general enough to cover this situation, and this kind of commutativity can fail even under very nice circumstances (see \cite{Nelson59} for a famous example).
This challenge has presented itself from the early days of Haag-Kastler's algebraic approach to QFT, and we do not offer a general solution.
However, we can use the standard subspace approach to vertex algebra nets $\cP(I)$ to give a conceptual proof that a simpler-looking statement is equivalent AQFT-locality.
Similar results were shown by Bisognano-Wichmann in the context considered in their original article \cite[Thm. 3(f)]{BisognanoWichmann75}.

Define a family of von Neumann subalgebras $\cC^\cH(I) \subset \cB(\cH)$ by $\cC^\cH(I) = \cA_\cV(I')'$.
Let $\cH_\cC \subset \cH$ be the Hilbert space
\begin{equation}\label{eqn: HC}
\cH_\cC = \overline{\left(\bigvee_{I \subset S^1} \cC^\cH(I)\right)\Omega},
\end{equation}
where $\bigvee_I \cC^\cH(I)$ is the von Neumann algebra generated by the algebras $\cC^\cH(I)$.
Let $p_\cC$ be the projection of $\cH$ onto $\cH_\cC$, and let $\cC(I) = \cC^\cH(I)p_\cC$, the cutdown of the net of algebras $\cC^\cH$ to its vacuum Hilbert space $\cH_\cC$.

We will show that the algebras $\cC(I)$ form a Haag-Kastler net on $\cH_\cC$, but it is convenient to first check the `one-particle' version.

\begin{lem}\label{lem: standard subspace net for dual net}
    The real subspaces $\overline{\cC(I)_{sa}\Omega}$ are a M\"obius-covariant net of standard subspaces of the Hilbert space $\cH_\cC$.
\end{lem}
\begin{proof}
    We verify that these subspaces satisfy Definition~\ref{defn: unitary net of subspaces}.
    Isotony of the net is clear, as is the spectrum condition.
    The algebras $\cC^\cH(I)$ are M\"obius covariant with respect to the given unitary representation on $\cH$. It follows that $\cH_\cC$ is invariant under the M\"obius group, and that the algebras $\cC(I)$ are M\"obius covariant, and indeed that the subspaces $\overline{\cC(I)_{sa}\Omega}$ are as well.
    
    We next check locality.
    By definition $\cC^\cH(I)$ commutes with $\cP(I')$, and therefore if $x \in \cC^\cH(I)_{sa}$ and $Y \in \cP(I')_{sa}$ we have $\ip{x\Omega,Y\Omega} \in \bbR$.
    Recall that $\hat \cK_n(I')$ is the norm closure of $\cP(I')_{sa}\Omega$ in $\cH$, and observe that it now follows that
    \[
    \cC(I)_{sa}\Omega = \cC^\cH(I)_{sa}\Omega \subset \hat \cK_n(I')' = \hat \cK_n(I),
    \]
    with the last equality by Proposition~\ref{prop: unitary K weakly haag dual}.
    The locality of the subspaces $\overline{\cC(I)_{sa}\Omega}$ now follows from the locality of the net $\hat \cK_n(I)$.
    
    It remains to check the cyclicity condition.
    Let $\tilde \cH_\cC \subset \cH_\cC$ be the closed span of the subspaces $\cC(I)\Omega$.
    Then evidently the given subspaces form a M\"obius-covariant net of subspaces of $\tilde \cH_\cC$.
    Then by the Reeh-Schlieder property \cite[Thm. 3.2.1]{LongoLectureNotesI} we have $\tilde \cH_\cC = \overline{\cC(I)\Omega}$ for each interval $I$.
    Hence $\tilde \cH_\cC$ is invariant under each algebra $\cC(I)$.
    It follows that $\tilde \cH_\cC$ is invariant under $\bigvee_I \cC(I)$, from which it follows that $\cH_\cC \subset \tilde \cH_\cC$.
    Hence $\cH_\cC = \tilde \cH_\cC$, which completes the proof of the cyclicity axiom.
\end{proof}

\begin{thm}\label{thm: commutant net locality}
Let $\cV$ be a unitary M\"obius vertex algebra with PCT operator $\theta$, let $\cC^\cH(I) = \cA_\cV(I')'$ be the dual net of von Neumann algebras, and let $\cC(I)$ be the cutdown of $\cC^\cH(I)$ to the vacuum Hilbert space $\cH_\cC$, as described above.
Then $\cC$ is a M\"obius-covariant Haag-Kastler net on $\cH_\cC$.
The modular conjugation associated with $\cC(I_+)$ and $\Omega$ is $\theta|_{\cH_\cC}$.
\end{thm}
\begin{proof}
The only axiom which needs to be checked is locality.
This will follow from the assertion regarding the modular conjugation, which we now verify.
Let $\cN = \overline{\cC(I_+)_{sa}\Omega}$, which is a standard subspace of $\cH_\cC$ by Lemma~\ref{lem: standard subspace net for dual net}.
Recall from the proof of Lemma~\ref{lem: standard subspace net for dual net} that $\cN \subset \hat \cK_n(I_+)$.
From the theory of M\"obius covariant nets of subspaces we know that the modular automorphism group of $\hat \cK_n(I_+)$ is $\Delta_+^{it} = U(v_{-2 \pi t})$ \cite[Thm. 3.3.1]{LongoLectureNotesI}.
As $\cN$ is invariant under $\Delta_+^{it}$, the modular automorphism group for $\cN$ is $\Delta_+^{it}|_{\cH_\cC}$ \cite[Cor. 2.1.8]{LongoLectureNotesI}.
Let $S_\cN$ be the involution on $\cH_\cC$ corresponding to $\cN$, and let $S_\cN = J_\cN\Delta_\cN^{1/2}$ be the polar decomposition.
For $\xi \in D(\Delta_\cN^{1/2})$, the map $\Delta_\cN^{it}\xi$ extends to a continuous function on the strip $\{-\tfrac12 \le \Im(\tau) \le 0\}$, holomorphic on the interior, mapping $\tau=-i/2$ to $\Delta^{1/2}_\cN\xi$.
As the same holds for $\Delta_+^{it}$, we have $\Delta^{1/2}_\cN\xi = \Delta^{1/2}_+\xi$ for $\xi \in \cN + i \cN$, and in particular for $\xi$ of the form $x\Omega$ with $x \in \cC(I_+)$.
Thus if $x \in \cC(I_+)$ we have
\[
J_\cN x^*\Omega = \Delta_\cN^{1/2} x \Omega = \Delta_+^{1/2}x\Omega = \theta x^* \Omega
\]
with the last equality by Theorem~\ref{thm: modular for unitary}.
Hence by Reeh-Schlieder (which in this case follows from Lemma~\ref{lem: standard subspace net for dual net}) we have $J_\cN = \theta|_{\cH_\cC}$, proving our claim about the modular conjugation for $\cC(I_+)$.

We now complete the proof of locality, for which it suffices to show that $\cC(I_+)$ and $\cC(I_-)$ commute.
By the main theorem of Tomita-Takesaki theory, we have $\theta \cC(I_+) \theta = \cC(I_+)'$.
On the other hand, since $\theta \cA_\cV(I_{\pm}) \theta = \cA_\cV(I_\mp)$, it follows that $\theta \cC^\cH(I_+) \theta = \cC^\cH(I_-)$.
Since $\theta$ leaves $\cH_\cC$ invariant we have $\theta \cC(I_+) \theta = \cC(I_-)$.
Hence $\cC(I_+)' = \cC(I_-)$, completing the proof of locality.
\end{proof}

\begin{cor}\label{cor: separating implies local}
    Let $\cV$ be a unitary M\"obius vertex algebra with vacuum vector $\Omega$ and PCT operator $\theta$. Then the following are equivalent:
    \begin{enumerate}[1)]
        \item $\cV$ is AQFT-local
        \item $\Omega$ is separating for $\cA_\cV(I)$, for some (or any) interval $I$
    \end{enumerate}
    If the above hold, then $\theta$ is the modular conjugation for $\cA(I_+)$ with respect to the vacuum state.
\end{cor}
\begin{proof}
    If $\cV$ is AQFT-local, then $\cA_\cV(I)$ is a M\"obius-covariant Haag-Kastler net, and $\Omega$ is separating for $\cA_\cV(I)$ by the Reeh-Schlieder property.
    Conversely, if $\Omega$ is separating, then it is cyclic for $\cC^\cH(I) = \cA_\cV(I')'$, and $\cH_\cC = \cH$.
    Hence by Theorem~\ref{thm: commutant net locality}, $\cC(I)$ is a M\"obius-covariant Haag-Kastler net on $\cH$, and by Haag duality for that net \cite[Thm. 2.19]{GabbianiFrohlich93} we have $\cC(I)=\cA_\cV(I)$, completing the proof.
    The computation of the modular conjugation follows from Theorem~\ref{thm: commutant net locality} as well.
\end{proof}
We note that the claim about the modular conjugation  could have alternatively been proven directly from \cite[Thm. B.6]{CKLW18}, under the assumption of energy bounds.
We also note that the proof of Corollary~\ref{cor: separating implies local} also shows that $\cV$ is AQFT-local when $\Omega$ is separating for $\bigcap_I \cA_\cV(I)$.

\def\lfhook#1{\setbox0=\hbox{#1}{\ooalign{\hidewidth
  \lower1.5ex\hbox{'}\hidewidth\crcr\unhbox0}}}

\end{document}